\RequirePackage{ifpdf}
\documentclass[12pt,a4paper,reqno]{amsart}
\usepackage{amsfonts,amsthm,latexsym,amsmath,amssymb,amscd}
\usepackage{graphics,graphicx}
\usepackage{microtype}
\usepackage[margin=2.6cm]{geometry}
\usepackage[colorlinks=true,linkcolor=blue,citecolor=blue]{hyperref}

\theoremstyle{plain}
\newtheorem{theorem}{Theorem}[section]
\newtheorem{proposition}[theorem]{Proposition}
\newtheorem{corollary}[theorem]{Corollary}
\newtheorem{lemma}[theorem]{Lemma}
\newtheorem{conjecture}[theorem]{Conjecture}

\newtheorem{THEO}{Theorem}

\newtheorem{THEOQ}{Theorem}

\theoremstyle{definition}

\newtheorem*{ack}{Acknowledgement}

\theoremstyle{remark}
\newtheorem{remark}[theorem]{Remark}
\newtheorem{example}[theorem]{Example}

\newcommand{\mathbbm}[1]{\mathbf{#1}}
\newcommand{\la}{\lambda}
\newcommand{\al}{\alpha}
\newcommand{\be}{\beta}
\newcommand{\si}{\sigma}
\newcommand{\eps}{\varepsilon}
\newcommand{\bC}{\mathbb C}
\newcommand{\bR}{\mathbb R}

\newcommand{\bE}{\mathbb E}
\newcommand{\bP}{\mathbb P}
\newcommand{\bCP}{\mathbb{CP}}
\newcommand{\bRP}{\mathbb{RP}}
\newcommand{\HH}{\mathcal H}

\newcommand{\cP}{\mathcal P}
\newcommand{\Ell}{\mathrm{Ell}}
\newcommand{\Disc}{\operatorname{Disc}}

\newcommand{\supp}{\operatorname{supp}}
\newcommand{\diam}{\operatorname{diam}}

\numberwithin{equation}{section}

\begin{document}

\title[Level crossings in random matrices III]
{Level crossings in random matrices III:\\ Analogues of Girko's circular law and Wigner's semicircle law}

\author[B.~Shapiro]{Boris Shapiro}
\address{Department of Mathematics, Stockholm University, SE-106 91 Stockholm, Sweden}
\email{shapiro@math.su.se}

\date{}



\begin{abstract}
For a random matrix pencil $A+\lambda B$, a \emph{level crossing} is a parameter
value $\lambda\in\mathbb{CP}^1$ at which the pencil has a multiple eigenvalue.
For a nondegenerate $n\times n$ pencil there are $n(n-1)$ such points, counted
with multiplicity. We study their empirical measure on the parameter sphere.

The main tool is an exact identity relating the chordal logarithmic potential of
the level-crossing measure to the normalized logarithmic energy of the spectrum of
$(A+\lambda B)/\sqrt{1+|\lambda|^2}$. For Gaussian Hermitian pencils the latter
matrix is an elliptic Ginibre matrix with parameter
$\tau(\lambda)=(1+\lambda^2)/(1+|\lambda|^2)$. We prove that the logarithmic
energy of the elliptic law is independent of $\tau$ and equals $-1/4$; thus the
circular and semicircular endpoints have the same energy.

This observation yields an unconditional large-$n$ theorem for independent
$GUE_n$ pencils: their empirical level-crossing measures converge in probability
to the uniform Fubini--Study measure on $\mathbb{CP}^1$. We also prove an exact
finite-$n$ symmetry theorem: whenever the law of $(A,B)$ is invariant under the
natural $SU(2)$ action, the full crossing process is rotation invariant and its
intensity is uniform, without Gaussian or moment assumptions. For real i.i.d.
pencils we reduce asymptotic uniformity to a scalar logarithmic-energy condition;
for real symmetric pencils we formulate the analogous conditional reduction.

Finally, we analyse diagonal deformation. The exactly diagonal Gaussian model has
Cauchy-distributed crossings, hence the uniform measure on $\mathbb{RP}^1$, and an
explicit singular energy profile. Monte Carlo experiments indicate a crossover
between the spherical and real-projective laws when the diagonal entries become of
order one. We state this crossover as a conjecture and give numerical evidence.
\end{abstract}

\maketitle

\section{Introduction}\label{s1}

\subsection{The problem}
Classical random matrix theory concerns the asymptotic distribution of the
eigenvalues of a single large random matrix; Wigner's semicircle law and Girko's
circular law are the two basic universality statements of this kind, see
\cite{Gi1,Gi2,Me,AGZ,BC}. The present paper is about a related but different object:
the distribution of the \emph{degeneracies} of a random one-parameter spectral
family
\begin{equation}\label{pencil}
        A+\la B,\qquad \la\in\bCP^1 .
\end{equation}
A \emph{level crossing} is a value of $\la$ for which $A+\la B$ has a multiple
eigenvalue. Thus the question is not where the eigenvalues of one random matrix
lie, but where the branch points of the spectral cover of the pencil sit in the
parameter sphere.

Such pencils occur throughout perturbation theory, with $A$ an unperturbed
operator, $B$ a perturbation and $\la$ the perturbation parameter \cite{Ka}; they
also arise in questions on the monodromy of spectra and on random Riemann surfaces,
cf.\ \cite{GP}.  After analytic continuation of a Hamiltonian parameter, generic
complex degeneracies are branch points of the eigenvalue sheets; their location
controls spectral monodromy and limits the domain of ordinary perturbation
expansions.  In random-matrix models of quantum spectra, the statistics of level
crossings and avoided crossings are therefore natural observables; see, for
example, \cite{ZVW}. The present paper continues \cite{ShZa1,ShZa2}.

Level crossings are the zeros of a discriminant. Writing
\[
        p_{A,B}(\la,t)=\det\bigl(A+\la B-tI\bigr),
        \qquad
        \Delta_{A,B}(\la)=\Disc_t\, p_{A,B}(\la,t),
\]
the level crossings are precisely the roots of $\Delta_{A,B}$, and
$\deg_\la\Delta_{A,B}=n(n-1)$. It is better to work homogeneously: for
$[\al:\be]\in\bCP^1$ put
\begin{equation}\label{eq:Dform}
        D_{A,B}(\al,\be):=\Disc_t\det\bigl(\al A+\be B-tI\bigr),
\end{equation}
a binary form of degree $N:=n(n-1)$. Whenever $D_{A,B}\not\equiv 0$ its zero
divisor consists of exactly $N$ points of $\bCP^1$ with multiplicity, and we let
\begin{equation}\label{eq:mun}
        \mu_{A,B}:=\frac1N\sum_{D_{A,B}(v)=0}\delta_{[v]}
\end{equation}
be the corresponding probability measure on $\bCP^1$, the \emph{empirical
level-crossing measure}. We stress that $\mu_{A,B}$ is a probability measure on a
\emph{compact} space for every realization; this simple observation will do real
work below.

Let $\si$ denote the rotation-invariant (Fubini--Study) probability measure on
$\bCP^1$; in the affine coordinate $\la=x+iy$,
\begin{equation}\label{eq:unif}
        d\si=\frac{dx\,dy}{\pi(1+|\la|^2)^2}.
\end{equation}
Following \cite{ShZa2} we call \eqref{eq:unif} the \emph{uniform law} for level
crossings. The basic exact result of \cite{ShZa2} is that complex Gaussian pencils
realize it identically in $n$.

\begin{THEOQ}[{\cite[Thm.~1, Prop.~1--2]{ShZa2}}]\label{th:GE}
Let $A,B$ be independent $n\times n$ complex Gaussian matrices with the same
entry variances in $A$ and $B$ (allowing a common variance on the diagonal different
from the off-diagonal variance). For every $n\ge2$, the expected empirical
level-crossing measure is exactly $\si$. The same conclusion holds for the
Gaussian subspace models treated in \cite{ShZa2}; see also the symmetry
formulation in \S\ref{sec:symmetry}.
\end{THEOQ}

The natural question, and the subject of this paper, is how far \eqref{eq:unif}
persists beyond the complex Gaussian case, and what replaces it when it fails.

\subsection{The mechanism}\label{ss:mechanism}
Everything here rests on one identity, which we now describe informally; see
\S\ref{sec:potential} for precise statements.

Let $\delta(p,q)\in[0,1]$ be the chordal distance on $\bCP^1$ and let
\[
        P^{\mu}(p):=\int_{\bCP^1}\log\delta(p,q)\,d\mu(q)
\]
be the chordal logarithmic potential of a probability measure $\mu$. Two elementary
facts drive the paper: $P^{\si}\equiv-\tfrac12$ is constant, and
\begin{equation}\label{eq:intro_laplace}
        \mu=\si+\frac1{2\pi}\Delta_\la P^{\mu},
\end{equation}
so that $\mu=\si$ if and only if $P^\mu$ is constant. Now normalize the pencil,
\[
        \widehat C_n(\la):=\frac{A_n+\la B_n}{\sqrt{1+|\la|^2}},
\]
and let $\zeta_1(\la),\dots,\zeta_n(\la)$ be its eigenvalues. Because the
discriminant is homogeneous of degree $N$ in the matrix entries, one gets the exact
identity, valid for every realization (Proposition~\ref{prop:key}),
\begin{equation}\label{eq:intro_key}
        P^{\mu_n}(\la)
        \;=\;
        \Lambda_n(\la)-\int_{\bCP^1}\Lambda_n\,d\si-\tfrac12,
        \qquad
        \Lambda_n(\la):=\frac1{N}\sum_{i\ne j}\log|\zeta_i(\la)-\zeta_j(\la)| .
\end{equation}
Thus \emph{the potential of the level-crossing measure is the normalized
logarithmic energy of the normalized spectrum}, up to a constant. All the
$\la$-dependence of the answer is carried by $\Lambda_n$.

For Gaussian Hermitian pencils, and more generally for the standard elliptic
Ginibre family, the normalized matrix $\widehat C_n(\la)$ is an elliptic random
matrix with parameter
\begin{equation}\label{eq:tau}
        \tau(\la)=\frac{1+\la^2}{1+|\la|^2},
        \qquad
        |\tau(\la)|^2=1-Y(\la)^2,
        \qquad
        Y(\la)=\frac{2\Im\la}{1+|\la|^2},
\end{equation}
$Y$ being the height coordinate on the sphere with respect to the axis through
$\la=\pm i$. In any model for which the limiting spectral measure is the elliptic
law $\Ell_{|\tau|}$, the corresponding limit of $\Lambda_n$ is therefore expected
to depend only on $Y$. Combining this with \eqref{eq:intro_laplace} gives the
\emph{master formula}. When the limiting energy profile $G$ is sufficiently smooth,
the absolutely continuous part of the limiting level-crossing measure has density
\begin{equation}\label{eq:intro_master}
        \varphi(Y)=1+2\,\frac{d}{dY}\Bigl[(1-Y^2)\,G'(Y)\Bigr]
\end{equation}
with respect to $\si$, where $G(Y)$ is the limiting energy profile. In general the
formula is interpreted distributionally; a kink in $G$ produces a singular measure
on the corresponding latitude. Precise statements are given in
Proposition~\ref{prop:master} and Theorem~\ref{th:criterion}.

\subsection{The key computation}
The family $\Ell_\tau$ interpolates between Girko's circular law ($\tau=0$) and
Wigner's semicircle law ($\tau=1$). Our main structural observation is:

\begin{THEO}\label{th:energy}
For every $\tau\in\bC$ with $|\tau|\le1$ the logarithmic energy of the elliptic
law is
\[
        I(\Ell_\tau)=\iint\log|z-w|\,d\Ell_\tau(z)\,d\Ell_\tau(w)=-\frac14 ,
\]
independently of $\tau$. More generally, if $\mu$ is a rotation-invariant
compactly supported probability measure on $\bC$ with finite logarithmic energy
and $T_\tau(z)=z+\tau\bar z$, then
$I\bigl((T_\tau)_*\mu\bigr)=I(\mu)+\log^+|\tau|$.
\end{THEO}

The proof is three lines (\S\ref{sec:energy}): $\Ell_\tau=(T_\tau)_*\Ell_0$, the
increment $T_\tau z-T_\tau w$ equals $s+\tau\bar s$ with $s=z-w$, the argument of
$s$ is uniform and independent of $|s|$ by rotation invariance, and
$\frac1{2\pi}\int_0^{2\pi}\log|\tau+e^{i\psi}|\,d\psi=\log^+|\tau|$ by Jensen's
formula.

Theorem~\ref{th:energy} says that $G$ in \eqref{eq:intro_master} is constant
whenever the limiting spectral law of $\widehat C_n(\la)$ is elliptic, so
$\varphi\equiv1$ and the limit is uniform. In particular the Girko-type and the
Wigner-type level-crossing laws coincide. This corrects the expectation, stated in
\cite{ShZa2} and in earlier drafts of the present work, that the Hermitian limit
should be a non-uniform, $SO(2)$-invariant measure. What is true is that the
finite-$n$ Hermitian laws are non-uniform --- for instance $GUE_2$ has density
$2|Y|$ with respect to $\si$ --- but the non-uniformity disappears in the limit. The numerical data indicate slow
convergence and are consistent with an $O(n^{-1}\log n)$ first correction; we do
not prove this rate here. See Figure~\ref{fig:cdf} and Remark~\ref{rem:rate}.

\subsection{Main results}
Beyond Theorems~\ref{th:GE} and \ref{th:energy}, we prove the following.

\smallskip\noindent
\textbf{(1) An exact law from symmetry, without Gaussian or moment hypotheses.}

\begin{THEO}\label{th:symmetry}
Let $n\ge2$, and let $(A,B)$ be a random pair of $n\times n$ complex matrices whose joint
distribution is invariant under the $SU(2)$-action
$\mathfrak U:(A,B)\mapsto(uA+vB,\,-\bar vA+\bar uB)$, $|u|^2+|v|^2=1$, and assume
$D_{A,B}\not\equiv0$ almost surely. Then the level-crossing point process on
$\bCP^1$ is invariant under the induced action of $SO(3)$. In particular
$\bE\mu_{A,B}=\si$ for every $n$, and all correlation measures of the crossing
process are rotation invariant.
\end{THEO}

This contains Theorem~\ref{th:GE} and considerably more: no moments need exist. For
example, if the entry pairs $(a_{ij},b_{ij})\in\bC^2$ are i.i.d.\ with an
$SU(2)$-invariant law --- Gaussian, uniform on a sphere, or heavy-tailed with an
arbitrary radial part --- and the resulting pencil is nondegenerate almost surely,
then its expected empirical crossing measure is exactly uniform on $\bCP^1$ for
every $n$. See \S\ref{sec:symmetry}.

\smallskip\noindent
\textbf{(2) An unconditional theorem in the Hermitian class.}

\begin{THEO}\label{th:GUEmain}
Let $A_n,B_n$ be independent $GUE_n$ matrices. Then the empirical level-crossing
measures $\mu_n$ of $A_n+\la B_n$ converge, in probability and weakly on $\bCP^1$,
to the uniform measure $\si$. The same holds when $A_n$ and $B_n$ are independently
drawn from any Gaussian ensembles for which $\widehat C_n(\la)$ is a complex
elliptic Ginibre matrix --- in particular for $GUE$/Ginibre mixed pencils.
\end{THEO}

 The proof is unconditional: the
uniform circular-law hypothesis (UCL) is not needed, because after normalization
the expected potential depends on $\la$ only through the scalar $|\tau(\la)|$ and
dominated convergence suffices; the logarithmic tail hypothesis (LT) becomes a
lemma via Schur's inequality; and the small-repulsion hypothesis (SR) is replaced
by a Christoffel-function bound $K_n(z,z)\le \frac n\pi e^{\sup\Delta V/4}$ valid
for any determinantal Coulomb gas, which we prove by a sub-mean-value argument
(Lemma~\ref{lem:christoffel}).

\smallskip\noindent
\textbf{(3) Conditional reductions for the real classes.}
For real i.i.d.\ pencils the global spectral input is the circular law, so the
remaining issue is the scalar convergence
\begin{equation}\label{eq:E}
        \bE\,\Lambda_n(\la)\longrightarrow -\tfrac14 ,
        \tag{E}
\end{equation}
for almost every fixed $\la$, together with a local uniform integrability bound.
For real symmetric Gaussian pencils the normalized matrix is complex symmetric;
we formulate the analogous reduction conditional also on the appropriate global
elliptic law.  In both cases the new criterion is substantially weaker than the
uniform circular-law, small-repulsion and logarithmic-tail hypotheses used in the
earlier approach. See Theorem~\ref{th:criterion} and \S\ref{sec:real}.

\smallskip\noindent
\textbf{(4) No separate hypothesis is needed to exclude mass on $\bRP^1$.}
For real ensembles crossings may approach the real projective line.  Once the
limiting measure off $\bRP^1$ has total mass one, however, compactness and
conservation of mass leave no room for an additional singular component on
$\bRP^1$ (Lemma~\ref{lem:mass}).

\smallskip\noindent
\textbf{(5) Diagonal deformation: a rigorous endpoint and a conjectural crossover.}
The diagonal model is exactly solvable.  If both matrices are real diagonal
Gaussians, every crossing is real and its one-point law is Cauchy; the empirical
crossing measure converges to the uniform measure on $\bRP^1$.  Its energy profile
is
\[
G_{\rm diag}(Y)=\mathrm{const}+\frac12\log\frac{1+|Y|}{2},
\]
whose kink at $Y=0$ produces precisely the equatorial singular mass in the
distributional master formula.  For a real symmetric pencil with the diagonal
scaled by $s$, the diagonal perturbation has operator norm of order
$s\sqrt{\log n/n}$, while a typical diagonal entry becomes order one when
$s\asymp\sqrt n$.  These observations identify the natural scales but, because
the normalized pencil is non-normal away from $\bRP^1$, they do not by themselves
prove stability of the level-crossing measure.  Our simulations at
$s=\varsigma\sqrt{n/2}$ show a clear family interpolating between the spherical and
equatorial laws.  We formulate this as Conjecture~\ref{conj:transition} rather
than as a theorem.

\subsection{Organization}
\S\ref{sec:potential} develops the chordal potential theory and proves the key
identity and the master formula. \S\ref{sec:symmetry} proves the exact symmetry
theorem, and \S\ref{sec:energy} computes the energy of the elliptic law.
\S\ref{sec:apriori} contains the a priori estimates used in the convergence
criterion. \S\ref{sec:hermitian} proves the unconditional $GUE$ theorem, while
\S\ref{sec:real} treats the conditional real classes. \S\ref{sec:transition}
analyses diagonal deformation, separating the rigorous diagonal endpoint from the
conjectural crossover. \S\ref{sec:numerics} reports the simulations, and
\S\ref{sec:problems} lists open problems.

\section{The discriminant form and chordal potential theory}\label{sec:potential}

\subsection{Level crossings as the divisor of a binary form}
Fix $n\ge2$ and set $N:=n(n-1)$. Let $\al A+\be B$ be the homogeneous pencil and
let $D_{A,B}$ be as in \eqref{eq:Dform}.

\begin{lemma}\label{lem:form}
$D_{A,B}$ is a binary form of degree $N$ in $(\al,\be)$. If
$\xi_1(\al,\be),\dots,\xi_n(\al,\be)$ denote the eigenvalues of $\al A+\be B$ then
$D_{A,B}=\prod_{i<j}(\xi_i-\xi_j)^2$.
\end{lemma}

\begin{proof}
The coefficients of the characteristic polynomial of $\al A+\be B$ are homogeneous
polynomials in $(\al,\be)$, and the discriminant of a monic polynomial is a
polynomial in its coefficients; the product formula is the definition of the
discriminant. Since each $\xi_i$ is homogeneous of degree $1$ as an algebraic
function of $(\al,\be)$, the product $\prod_{i<j}(\xi_i-\xi_j)^2$ is homogeneous of
degree $2\binom n2=N$.
\end{proof}

We say the pencil is \emph{nondegenerate} if $D_{A,B}\not\equiv0$, equivalently if
$\al A+\be B$ has simple spectrum for some (hence generic) $(\al,\be)$. In that
case \eqref{eq:mun} defines a probability measure on $\bCP^1$. For the Gaussian and other absolutely continuous ensembles considered below,
nondegeneracy holds almost surely because $\{D\equiv0\}$ is a proper algebraic
subset of the parameter space.  In the general symmetry theorem we impose
nondegeneracy explicitly, since an invariant law may have atoms.

In the affine chart $[\al:\be]=[1:\la]$ we write $\Delta_n(\la):=D_{A,B}(1,\la)$,
so $\Delta_n$ is a polynomial of degree $\le N$ in $\la$, of degree exactly $N$
unless some crossings sit at $\la=\infty$. Working with the form rather than with
$\Delta_n$ removes all bookkeeping about crossings at infinity and about the total
mass.

\subsection{The chordal potential}
For $p,q\in\bCP^1$ with unit representatives $v,w\in\bC^2$ put
\[
        \delta(p,q):=|\det(v,w)|\in[0,1],
\]
independent of the choice of unit representatives. In the affine coordinate,
\begin{equation}\label{eq:chordal}
        \delta(\la,\mu)=\frac{|\la-\mu|}{\sqrt{(1+|\la|^2)(1+|\mu|^2)}} .
\end{equation}
This is one half of the Euclidean chordal distance under the stereographic
identification $\bCP^1\cong S^2$ with the unit sphere. For a probability measure
$\mu$ on $\bCP^1$ define
\[
        P^{\mu}(p):=\int_{\bCP^1}\log\delta(p,q)\,d\mu(q)\in[-\infty,0].
\]

\begin{lemma}\label{lem:pot_basic}
Let $\mu$ be a probability measure on $\bCP^1$. Then:
\begin{enumerate}
\item[(i)] $P^{\si}\equiv-\tfrac12$;
\item[(ii)] $\displaystyle\int_{\bCP^1}P^{\mu}\,d\si=-\tfrac12$;
\item[(iii)] in the affine chart, as distributions on $\bC$,
$\displaystyle \frac1{2\pi}\Delta_\la P^{\mu}=\mu-\si$; in particular $\mu=\si$ if
and only if $P^\mu$ is ($\si$-a.e.) constant;
\item[(iv)] for every $1\le r<\infty$ there is $C_r<\infty$, independent of $\mu$,
with $\int|P^{\mu}|^r\,d\si\le C_r$.
\end{enumerate}
\end{lemma}

\begin{proof}
(i) By rotation invariance we may take $p=[1:0]$, i.e.\ $\la_p=0$. Using
\eqref{eq:chordal} and \eqref{eq:unif}, and substituting $s=r^2$ and then
$u=s/(1+s)$,
\[
        P^{\si}=\int_{\bC}\log\frac{|\mu|}{\sqrt{1+|\mu|^2}}\,
        \frac{dx\,dy}{\pi(1+|\mu|^2)^2}
        =\frac12\int_0^{\infty}\log\frac{s}{1+s}\,\frac{ds}{(1+s)^2}
        =\frac12\int_0^1\log u\,du=-\frac12 .
\]
(ii) By Fubini and (i), $\int P^{\mu}d\si=\iint\log\delta\,d\si\,d\mu
=\int P^{\si}d\mu=-\tfrac12$.

(iii) By \eqref{eq:chordal},
\[
        P^{\mu}(\la)=\int\log|\la-\mu|\,d\mu(\mu)
        -\tfrac12\log(1+|\la|^2)
        -\tfrac12\int\log(1+|\mu|^2)\,d\mu(\mu),
\]
and $\frac1{2\pi}\Delta_\la\int\log|\la-\mu|d\mu=\mu$, while
$\frac1{2\pi}\Delta_\la\bigl[\tfrac12\log(1+|\la|^2)\bigr]
=\pi^{-1}(1+|\la|^2)^{-2}$ is the density of $\si$. The last statement follows
since $\Delta P^\mu=0$ forces $\mu=\si$, and conversely by (i).

(iv) By Jensen's inequality applied to $|\cdot|^r$ and Fubini,
\[
 \int|P^{\mu}|^r\,d\si
 \le \iint|\log\delta(p,q)|^r\,d\si(p)d\mu(q)=C_r,
\]
where $C_r=\int|\log\delta(p,q)|^r\,d\si(p)$ is finite and independent of $q$ by
rotation invariance.
\end{proof}

Part (iii) is the form of the Poincar\'e--Lelong formula adapted to the sphere: it
says that a probability measure on $\bCP^1$ is uniform precisely when its chordal
potential is constant, and it exhibits the deviation from uniformity as a
Laplacian. Part (iv) gives, for free, uniform integrability of the potentials of an
arbitrary family of probability measures; we use it to convert weak convergence of
measures into $L^1$ convergence of potentials.

\begin{lemma}[Weak convergence and potentials]\label{lem:weakL1}
If $\mu_k\to\mu$ weakly on $\bCP^1$ then $P^{\mu_k}\to P^{\mu}$ in $L^1(\si)$.
Conversely, if $P^{\mu_k}\to \Pi$ in $L^1(\si)$ then $\mu_k\to\si+\frac1{2\pi}
\Delta\Pi$ weakly, and $\Pi=P^{\mu}$ for the limit $\mu$.
\end{lemma}

\begin{proof}
For $0<\eps<1$ put $\log_\eps t:=\max(\log t,\log\eps)$ and
$P^{\mu}_\eps(p):=\int\log_\eps\delta(p,q)\,d\mu(q)$. The kernel $\log_\eps\delta$
is bounded and continuous, so $P^{\mu_k}_\eps\to P^{\mu}_\eps$ pointwise and
boundedly, hence in $L^1(\si)$. Moreover $0\le P^{\mu}_\eps-P^{\mu}$ for every
$\mu$, and by Fubini and rotation invariance
\[
        \int\bigl(P^{\mu}_\eps-P^{\mu}\bigr)d\si
        =\iint\bigl(\log_\eps\delta-\log\delta\bigr)d\si(p)\,d\mu(q)=c_\eps,
\]
where $c_\eps:=\int(\log_\eps\delta(p,q)-\log\delta(p,q))\,d\si(p)$ does not depend
on $q$ and satisfies $c_\eps\downarrow0$ as $\eps\downarrow0$ by monotone
convergence, since $P^{\si}=-\tfrac12$ is finite. Hence
\[
        \|P^{\mu_k}-P^{\mu}\|_{L^1(\si)}
        \le c_\eps+\|P^{\mu_k}_\eps-P^{\mu}_\eps\|_{L^1(\si)}+c_\eps
        \xrightarrow[k\to\infty]{}2c_\eps ,
\]
and letting $\eps\downarrow0$ proves the first statement. The converse is immediate
from Lemma~\ref{lem:pot_basic}(iii), since $L^1$ convergence implies distributional
convergence of Laplacians.
\end{proof}

\subsection{The key identity}
We now relate $P^{\mu_n}$ to the spectrum of the normalized pencil. Assume the
entries of $A_n,B_n$ have variance $1/n$, so that for a unit vector
$v=(\al,\be)\in\bC^2$ the matrix $\al A_n+\be B_n$ is normalized in the usual way;
in the affine chart $v=(1,\la)/\sqrt{1+|\la|^2}$ this is exactly
$\widehat C_n(\la)=(A_n+\la B_n)/\sqrt{1+|\la|^2}$.

\begin{proposition}[Key identity]\label{prop:key}
Let $(A,B)$ be a nondegenerate pencil, let $\mu_n=\mu_{A,B}$ and let
\[
        \Lambda_n(p):=\frac1N\sum_{i\ne j}\log\bigl|\zeta_i(v)-\zeta_j(v)\bigr|,
        \qquad v\ \text{a unit representative of }p,
\]
where $\zeta_1(v),\dots,\zeta_n(v)$ are the eigenvalues of $\al A+\be B$. Then
$\Lambda_n$ is well defined on $\bCP^1$ and, for every realization,
\begin{equation}\label{eq:key}
        P^{\mu_n}(p)=\Lambda_n(p)-\int_{\bCP^1}\Lambda_n\,d\si-\frac12 .
\end{equation}
Equivalently, in the affine chart,
\begin{equation}\label{eq:keyaffine}
        \frac1N\log|\Delta_n(\la)|
        =\frac12\log(1+|\la|^2)+\Lambda_n(\la),
        \qquad
        \Lambda_n(\la)=\frac1N\sum_{i\ne j}\log|\zeta_i(\la)-\zeta_j(\la)| ,
\end{equation}
with $\zeta_i(\la)$ the eigenvalues of $\widehat C_n(\la)$.
\end{proposition}

\begin{proof}
Replacing $v$ by $e^{i\theta}v$ multiplies all $\zeta_i$ by $e^{i\theta}$ and does
not change $\Lambda_n$, so $\Lambda_n$ is well defined on $\bCP^1$. By
Lemma~\ref{lem:form}, $|D_{A,B}(v)|=\prod_{i<j}|\zeta_i(v)-\zeta_j(v)|^2$, whence
\begin{equation}\label{eq:DLambda}
        \frac1N\log|D_{A,B}(v)|=\Lambda_n(p).
\end{equation}
On the other hand, factoring the binary form over $\bC$ and normalizing the roots,
$D_{A,B}(v)=c\prod_{k=1}^N\det(v,v_k)$ with unit representatives $v_k$ of the
crossings $q_k$ and some constant $c\ne0$, so that
\begin{equation}\label{eq:DP}
        \frac1N\log|D_{A,B}(v)|=\frac1N\log|c|+\frac1N\sum_k\log\delta(p,q_k)
        =\frac1N\log|c|+P^{\mu_n}(p).
\end{equation}
Comparing \eqref{eq:DLambda} and \eqref{eq:DP} gives
$P^{\mu_n}=\Lambda_n-\frac1N\log|c|$; integrating against $\si$ and using
Lemma~\ref{lem:pot_basic}(ii) identifies the constant, which yields \eqref{eq:key}.
For \eqref{eq:keyaffine}, take $v=(1,\la)/\sqrt{1+|\la|^2}$ and use homogeneity:
$D_{A,B}(v)=(1+|\la|^2)^{-N/2}\Delta_n(\la)$.
\end{proof}

Identity \eqref{eq:key} is the exact form of the heuristic in \S\ref{ss:mechanism}.
It has two immediate consequences worth recording. First, the additive constant in
\eqref{eq:key} is not a nuisance: it is determined by $\si$-averaging, so no
normalization is left undetermined. Second, since the left-hand side of
\eqref{eq:key} determines $\mu_n$ by Lemma~\ref{lem:pot_basic}(iii), \emph{any}
asymptotic information about the scalar function $\Lambda_n$ translates directly
into information about level crossings.

\subsection{Conservation of mass}
The next lemma is the reason we may ignore the real projective line.

\begin{lemma}[No hidden mass]\label{lem:mass}
Let $\nu_k$ be probability measures on $\bCP^1$, let $\Sigma\subset\bCP^1$ be
closed with $\si(\Sigma)=0$, and let $\Psi$ be a probability measure with
$\Psi(\Sigma)=0$. If $\int f\,d\nu_k\to\int f\,d\Psi$ for every continuous $f$
supported in $\bCP^1\setminus\Sigma$, then $\nu_k\to\Psi$ weakly on $\bCP^1$.
\end{lemma}

\begin{proof}
By compactness of $\bCP^1$ every subsequence has a weakly convergent
sub-subsequence, say $\nu_{k_j}\to\nu$. For $f$ continuous with support in the open
set $U=\bCP^1\setminus\Sigma$ we get $\int f d\nu=\int f d\Psi$, hence
$\nu|_U=\Psi|_U$. Since $\Psi(U)=1$ and $\nu$ is a probability measure,
$\nu(U)=1$ forces $\nu(\Sigma)=0$ and $\nu=\Psi$. All subsequential limits agree,
so $\nu_k\to\Psi$.
\end{proof}

In particular, if the limiting density is known to be uniform off $\bRP^1$, then it
is uniform, and no separate hypothesis excluding an $O(n^2)$ cloud of almost-real
crossings is required. This removes the hypothesis
$\lim_{\eps\downarrow0}\limsup_n\bE N_n(\eps)/n(n-1)=0$ from all statements below.

\subsection{The master formula}
It remains to record how a limiting energy profile determines the limiting measure.
Recall $Y=Y(\la)$ from \eqref{eq:tau}; $Y$ is the height coordinate on
$S^2\cong\bCP^1$ along the axis through $\la=\pm i$, its distribution under $\si$
is uniform on $[-1,1]$, and the spherical Laplacian of a function of $Y$ alone is
$\Delta_{S^2}G=\partial_Y\bigl[(1-Y^2)\partial_YG\bigr]$.

\begin{proposition}[Master formula]\label{prop:master}
Suppose $\Lambda_n\to G$ in $L^1(\si)$ in probability, where $G$ is a
deterministic function of $Y$ alone. Then $\mu_n\to\Psi$ weakly in probability on
$\bCP^1$, where $\Psi$ is axisymmetric.  If $m:=Y_*\Psi$ denotes its marginal on
$[-1,1]$, then, in the sense of distributions on $(-1,1)$,
\begin{equation}\label{eq:masterdist}
        dm(Y)=\frac12\,dY
        +d\!\left[(1-Y^2)G'(Y)\right].
\end{equation}
Thus a jump of $(1-Y^2)G'(Y)$ contributes an atom of the same mass to the
$Y$-marginal, uniformly distributed along the corresponding latitude.

If $G\in C^2(-1,1)$, the absolutely continuous part of $\Psi$ on
$\{-1<Y<1\}$ has density
\begin{equation}\label{eq:master}
        \varphi(Y)=1+2\,\frac{d}{dY}\Bigl[(1-Y^2)G'(Y)\Bigr]
\end{equation}
with respect to $\si$.

If, in addition, $G$ is even and locally $C^1$ on $(0,1)$, then at every point
$0<y<1$ where the one-sided derivative is defined,
\begin{equation}\label{eq:cdf}
        \Psi\{|Y|\le y\}=y+2(1-y^2)G'(y).
\end{equation}
The formula includes any singular mass at $Y=0$. In particular, for a globally
smooth profile, $G$ is constant if and only if $\Psi=\si$.
\end{proposition}

\begin{proof}
By \eqref{eq:key}, $P^{\mu_n}\to G-\int G\,d\si-\tfrac12$ in $L^1(\si)$ in
probability, and Lemma~\ref{lem:weakL1} gives
$\mu_n\to\si+\frac1{2\pi}\Delta_\la G$. Passing from the affine Laplacian to the
spherical one,
$\Delta_\la=4(1+|\la|^2)^{-2}\Delta_{S^2}$ and
$dx\,dy=(1+|\la|^2)^2dA/4$, hence
\[
\frac1{2\pi}\Delta_\la G\,dx\,dy
   =\frac1{2\pi}\Delta_{S^2}G\,dA
   =2\Delta_{S^2}G\,d\si .
\]
For a function of $Y$,
$\Delta_{S^2}G=\partial_Y[(1-Y^2)G']$ distributionally.  Since the
$Y$-marginal of $\si$ is $\tfrac12dY$, pushing the preceding identity forward by
$Y$ gives \eqref{eq:masterdist}.  Under $C^2$ regularity this is equivalent to
\eqref{eq:master}.  If $G$ is even, then $(1-Y^2)G'(Y)$ is odd away from its
possible jumps; integrating \eqref{eq:masterdist} over $[-y,y]$ gives
\eqref{eq:cdf}.
\end{proof}

\begin{remark}[Finite-$n$ mean law]\label{rem:finite_master}
Assume $\Lambda_n$ is integrable. Taking expectations in the key identity gives
the exact finite-$n$ relation
\[
 \bE\mu_n=\si+\frac1{2\pi}\Delta_\la\,\bE\Lambda_n
\]
in the sense of distributions.  In an axisymmetric ensemble for which
$g_n(Y):=\bE\Lambda_n$ is even, this becomes the same one-dimensional master
formula as above.  In particular, at points $0<y<1$ where $g_n$ is differentiable,
\[
\bE\mu_n\{|Y|\le y\}=y+2(1-y^2)g_n'(y),
\]
with any atom at $Y=0$ automatically included.  This is the identity used to
reconstruct the finite-$n$ curves in \S\ref{sec:numerics}; it is not an
asymptotic approximation.
\end{remark}

\begin{example}\label{ex:GUE2}
For $A,B$ independent $GUE_2$ matrices, \cite[Thm.~3]{ShZa2} gives the exact
density $\frac4\pi|y|(1+|\la|^2)^{-3}dx\,dy$, i.e.\ $\varphi(Y)=2|Y|$ with respect
to $\si$, and $\Psi\{|Y|\le y\}=y^2$. Inverting \eqref{eq:cdf} yields the exact
energy profile of the $2\times2$ complex elliptic Ginibre ensemble:
\begin{align}\label{eq:g2}
 G_2(Y)&=G_2(0)-\tfrac12\bigl[\,|Y|-\log(1+|Y|)\,\bigr],\\
 G_2(1)-G_2(0)&=-\tfrac12(1-\log2)=-0.153426\ldots .
\end{align}
This prediction is confirmed by simulation to four digits (\S\ref{sec:numerics}),
which is a stringent test of the whole framework.
\end{example}

\section{Exact uniform law from symmetry}\label{sec:symmetry}

Consider the $SU(2)$-action on pairs of matrices used in \cite{ShZa2}: for
$\mathfrak U=\begin{pmatrix} u&-\bar v\\ v&\bar u\end{pmatrix}\in SU(2)$,
\begin{equation}\label{eq:action}
        \mathfrak U:(A,B)\longmapsto (uA+vB,\;-\bar vA+\bar uB).
\end{equation}

\begin{lemma}[Equivariance]\label{lem:equivariance}
Write $g=\begin{pmatrix} u&v\\ -\bar v&\bar u\end{pmatrix}\in SU(2)$ and let
$(A',B')=\mathfrak U(A,B)$. Then, for row vectors $w=(\al,\be)$,
\[
        D_{A',B'}(w)=D_{A,B}(wg).
\]
Consequently $\mu_{A',B'}=(R_g)_*\mu_{A,B}$, where $R_g[w]=[wg^{-1}]$ is the
rotation of $\bCP^1\cong S^2$ induced by $g$.
\end{lemma}

\begin{proof}
Directly,
\[
        \al A'+\be B'=\al(uA+vB)+\be(-\bar vA+\bar uB)
        =(\al u-\be\bar v)A+(\al v+\be\bar u)B,
\]
and $(\al u-\be\bar v,\ \al v+\be\bar u)=wg$. Hence the pencils
$\{\al A'+\be B'\}$ and $\{\al A+\be B\}$ coincide after the substitution
$w\mapsto wg$, and the same is true of their discriminants. Therefore
$D_{A',B'}(w)=0$ iff $D_{A,B}(wg)=0$, with matching multiplicities, so the crossing
divisor of $(A',B')$ is the image of that of $(A,B)$ under $[w]\mapsto[wg^{-1}]$.
The action of $SU(2)$ on $\bCP^1$ by such substitutions is exactly the action of
$SO(3)=PSU(2)$ by rotations in the identification $\bCP^1\cong S^2$.
\end{proof}

\begin{proof}[Proof of Theorem~\ref{th:symmetry}]
By hypothesis $(A',B')=\mathfrak U(A,B)$ has the same law as $(A,B)$; by
Lemma~\ref{lem:equivariance} the random measure $(R_g)_*\mu_{A,B}$ has the same law
as $\mu_{A,B}$, for every $g\in SU(2)$. Hence the level-crossing point process is
invariant in law under all rotations of $S^2$. In particular its intensity measure
$\bE\mu_{A,B}$ is a rotation-invariant probability measure on $S^2$, and the only
such measure is $\si$. The statement about correlation measures is the same
argument applied to the factorial moment measures.
\end{proof}

Theorem~\ref{th:symmetry} is a purely structural statement: it uses no moment
assumption, no independence of entries, and no Gaussian input. We now list what it
covers.

\begin{corollary}[Non-Gaussian exact uniform laws]\label{cor:examples}
Assume in each case below that the pencil is nondegenerate almost surely. Then
$\bE\mu_{A,B}=\si$ for every $n\ge2$.
\begin{enumerate}
\item[(i)] The entry pairs $\bigl((a_{ij},b_{ij})\bigr)_{i,j}$ are i.i.d.\
$\bC^2$-valued random vectors whose common law is invariant under the standard
action of $SU(2)$ on $\bC^2$. Equivalently, in polar form
$(a_{ij},b_{ij})=R_{ij}\,\Theta_{ij}$ with $\Theta_{ij}$ uniform on the unit sphere
$S^3\subset\bC^2$ and $R_{ij}\ge0$ arbitrary, independent of $\Theta_{ij}$. No
moments are required; $R_{ij}$ may be heavy tailed.
\item[(ii)] More generally, the entry pairs need only be jointly
$SU(2)$-invariant, the same $\mathfrak U$ acting simultaneously in all coordinates;
independence across $(i,j)$ is irrelevant.
\item[(iii)] Any linear pattern imposed identically on $A$ and $B$: complex
symmetric, Toeplitz, Hankel, banded, banded Toeplitz, circulant or diagonal
matrices, or matrices supported on any fixed set of entries, provided the surviving
entry pairs satisfy (ii). This recovers and extends
\cite[Prop.~2]{ShZa2}.
\item[(iv)] $A,B$ i.i.d.\ complex Ginibre with an arbitrary common variance on the
diagonal, recovering Theorem~\ref{th:GE}.
\end{enumerate}
\end{corollary}

\begin{proof}
In each case the action \eqref{eq:action} acts on the pair of matrices
entrywise: $(a_{ij},b_{ij})\mapsto(ua_{ij}+vb_{ij},\,-\bar va_{ij}+\bar ub_{ij})$,
i.e.\ by the fixed matrix $\mathfrak U$ applied to the column vector
$(a_{ij},b_{ij})^{\mathsf T}\in\bC^2$. Invariance of the joint law of the entry
pairs therefore gives invariance of the law of $(A,B)$, and
Theorem~\ref{th:symmetry} applies. For (i) note that an $SU(2)$-invariant law on
$\bC^2$ is exactly one whose angular part is uniform on $S^3$, since $SU(2)$ acts
transitively on $S^3$. For (iii), the pattern is preserved because $\mathfrak U$
acts within each coordinate pair separately. For (iv), a standard complex Gaussian
vector in $\bC^2$ is $U(2)$-invariant.
\end{proof}

\begin{remark}
Theorem~\ref{th:symmetry} explains \emph{why} the complex Gaussian answer of
\cite{ShZa2} is independent of $n$ and of the variance: those are not features of
the Gaussian computation but of a symmetry. It also clarifies the notion of a
subspace \emph{adjusted to the $SU(2)$-action}: the content of that condition is
precisely that the induced measure on $W\times W$ stays invariant.
\end{remark}

\begin{remark}[Real ensembles keep a circle of symmetry]\label{rem:SO2}
If $(A,B)$ is invariant only under the real subgroup $SO(2)\subset SU(2)$, i.e.\
under $(A,B)\mapsto(\cos\theta\,A+\sin\theta\,B,\,-\sin\theta\,A+\cos\theta\,B)$,
then by the same argument $\bE\mu_{A,B}$ is invariant under the corresponding
one-parameter group of rotations of $S^2$. The fixed points of that group are
$\la=\pm i$; the invariant coordinate is precisely $Y$ of \eqref{eq:tau}. This
applies to any pencil of i.i.d.\ real Gaussian ensembles --- $GOE$, real Ginibre,
real symmetric with scaled diagonal --- because the entry pairs
$(a_{ij},b_{ij})\in\bR^2$ are then standard real Gaussian vectors, hence
$SO(2)$-invariant. Combined with the deterministic symmetry
$\Delta_n(\bar\la)=\overline{\Delta_n(\la)}$, valid whenever $A,B$ are real or
Hermitian, we conclude that for all these ensembles $\bE\mu_n$ is a function of
$|Y|$ alone. This proves the claim, stated without proof in \cite{ShZa2}, that
$\cP^{GUE}_n(\theta,\phi)$ does not depend on the azimuthal angle.
\end{remark}

\section{The logarithmic energy of the elliptic law}\label{sec:energy}

Recall that the \emph{elliptic law} $\Ell_\tau$, $\tau\in[0,1]$, is the uniform
probability measure on the ellipse with semi-axes $1+\tau$ and $1-\tau$; for
$\tau\in\bC$, $|\tau|\le1$, one sets $\Ell_\tau:=(e^{i\theta})_*\Ell_{|\tau|}$
where $\tau=|\tau|e^{2i\theta}$. Thus $\Ell_0$ is Girko's circular law and $\Ell_1$
is Wigner's semicircle law on $[-2,2]$. Write
$I(\mu)=\iint\log|z-w|\,d\mu(z)d\mu(w)$.

\begin{proof}[Proof of Theorem~\ref{th:energy}]
Let $\mu$ be rotation invariant on $\bC$ with $I(\mu)>-\infty$, let
$T_\tau(z)=z+\tau\bar z$ and let $Z,W$ be i.i.d.\ with law $\mu$. Since
$T_\tau Z-T_\tau W=S+\tau\bar S$ with $S:=Z-W$, and since the law of $S$ is
rotation invariant, we may write $S=Re^{i\Theta}$ with $\Theta$ uniform on
$[0,2\pi)$ and independent of $R=|S|$. Then
$|S+\tau\bar S|=R\,|e^{i\Theta}+\tau e^{-i\Theta}|=R\,|e^{2i\Theta}+\tau|$, so
\[
        I\bigl((T_\tau)_*\mu\bigr)
        =\bE\log R+\bE\log|\tau+e^{2i\Theta}|
        =I(\mu)+\frac1{2\pi}\int_0^{2\pi}\log|\tau+e^{i\psi}|\,d\psi
        =I(\mu)+\log^+|\tau|,
\]
the last step by Jensen's formula applied to the polynomial $z\mapsto z+\tau$ on
the unit circle.

Now take $\mu=\Ell_0$, the normalized area measure of the unit disc. The real
linear map $T_\tau$ has constant Jacobian and takes the unit disc onto the ellipse
with semi-axes $1+\tau$ and $1-\tau$, hence $(T_\tau)_*\Ell_0=\Ell_\tau$ for
$\tau\in[0,1]$. Finally $I(\Ell_0)=-\tfrac14$: the logarithmic potential of the
uniform measure on the unit disc is $\log|z|$ for $|z|\ge1$ and
$\tfrac12(|z|^2-1)$ for $|z|\le1$, so
$I(\Ell_0)=\tfrac12\bigl(\int|z|^2d\Ell_0-1\bigr)=\tfrac12(\tfrac12-1)=-\tfrac14$.
Rotation of the plane does not change logarithmic energy, so the complex case
follows.
\end{proof}

\begin{remark}
It is instructive to check the two endpoints independently. For the semicircle law
$\mu_{sc}$ on $[-2,2]$ one has $\int\log|x-y|d\mu_{sc}(y)=\tfrac{x^2}4-\tfrac12$
on $[-2,2]$, so $I(\mu_{sc})=\tfrac14\cdot1-\tfrac12=-\tfrac14$, matching
$I(\Ell_0)$. The coincidence of the logarithmic energies of Girko's and Wigner's
laws is the analytic reason for the coincidence of the two level-crossing limit
laws.
\end{remark}

We shall also need the corresponding statement for Gaussian clouds, which governs
the diagonal-dominated regime of \S\ref{sec:transition}.

\begin{lemma}[Energy of an elliptic Gaussian cloud]\label{lem:gaussenergy}
Let $\gamma_\tau$, $\tau\in[0,1]$, be the centered complex Gaussian law on $\bC$
with $\bE|\zeta|^2=1$ and $\bE\zeta^2=\tau$. Then
\[
        I(\gamma_\tau)=I(\gamma_0)+\tfrac12\log\frac{1+\sqrt{1-\tau^2}}2 .
\]
\end{lemma}

\begin{proof}
Write $\zeta=a\eta+b\bar\eta$ with $\eta$ standard complex Gaussian
($\bE|\eta|^2=1$, $\bE\eta^2=0$) and $a\ge b\ge0$. Then $a^2+b^2=1$ and $2ab=\tau$,
so $a=\tfrac12(\sqrt{1+\tau}+\sqrt{1-\tau})$ and $b/a\le1$. Since
$\zeta=a\,T_{b/a}(\eta)$ and $\eta$ has a rotation-invariant law,
Theorem~\ref{th:energy} gives
$I(\gamma_\tau)=I(\gamma_0)+\log a+\log^+(b/a)=I(\gamma_0)+\log a$, and
$a^2=\tfrac12\bigl(1+\sqrt{1-\tau^2}\bigr)$.
\end{proof}

Note the contrast with Theorem~\ref{th:energy}: for the elliptic \emph{law} the
energy is constant, whereas for the elliptic Gaussian \emph{cloud} it is not. In
the language of \eqref{eq:tau}, $\sqrt{1-\tau^2}=|Y|$, so
Lemma~\ref{lem:gaussenergy} gives the profile
$G_{\rm diag}(Y)=\mathrm{const}+\tfrac12\log\frac{1+|Y|}2$ used in the exactly diagonal endpoint, Theorem~\ref{th:diag_endpoint}.

\section{A priori estimates}\label{sec:apriori}

Throughout this section $\widehat C_n=\widehat C_n(\la)$ is an $n\times n$ random
matrix with eigenvalues $\zeta_1,\dots,\zeta_n$,
\[
        L_n:=\frac1n\sum_k\delta_{\zeta_k},
        \qquad
        M_n:=\frac1{N}\sum_{i\ne j}\delta_{(\zeta_i,\zeta_j)},
        \qquad
        \Lambda_n=\iint\log|z-w|\,dM_n(z,w),
\]
and for $0<\eps<1<R$ we use the two-sided truncation
\[
        \Phi_{\eps,R}(z,w):=\max\bigl\{\log\eps,\ \min(\log|z-w|,\log R)\bigr\},
\]
a bounded continuous function on $\bC^2$ with values in $[\log\eps,\log R]$.

\subsection{The tail is universal}
The hypothesis (LT) of the earlier treatment is not needed: it is a consequence of
Schur's inequality and a second-moment bound on the entries.

\begin{lemma}[Universal logarithmic tail bound]\label{lem:tail}
Let $T_n:=\frac1n\|\widehat C_n\|_{HS}^2$ and suppose $\bE T_n^2\le K$. Then for
all $R\ge2$
\[
        \frac1{N}\,\bE\sum_{i\ne j}
        \mathbbm 1_{\{\max(|\zeta_i|,|\zeta_j|)>R\}}\,\log(2+|\zeta_i|+|\zeta_j|)
        \;\le\;\frac{C(\sqrt K+K^{3/4})}{R},
\]
with an absolute constant $C$. In particular, for each fixed bound on $K$ the tail is
$O(R^{-1})$, uniformly in $n$.  The hypothesis holds locally uniformly in $\la$
for the ensembles considered below with uniformly bounded fourth moments; for
Gaussian ensembles $\bE T_n^2=O(1)$.
\end{lemma}

\begin{proof}
By Schur's inequality $\sum_k|\zeta_k|^2\le\|\widehat C_n\|_{HS}^2=nT_n$. Using
$\log(2+s+t)\le\log(2+s)+\log(2+t)$ and symmetry, the left-hand side is at most
\begin{align*}
 &\frac2{N}\bE\sum_{i\ne j}\mathbbm 1_{\{|\zeta_i|>R\}}
 \bigl[\log(2+|\zeta_i|)+\log(2+|\zeta_j|)\bigr]\\
 &\qquad\le 2\,\bE\Biggl[
 \frac1n\sum_i\mathbbm 1_{\{|\zeta_i|>R\}}\log(2+|\zeta_i|)
 +\frac{\#\{i:|\zeta_i|>R\}}{n}
   \frac1{n-1}\sum_j\log(2+|\zeta_j|)\Biggr].
\end{align*}
For $s\ge R\ge2$ we have $\log(2+s)\le 2s$, so the first term is bounded by
$\frac2n\sum_i\mathbbm 1_{\{|\zeta_i|>R\}}|\zeta_i|\le\frac2{Rn}\sum_i|\zeta_i|^2
\le 2T_n/R$. For the second, Chebyshev gives $\#\{|\zeta_i|>R\}/n\le T_n/R^2$,
while by concavity of the logarithm and Cauchy--Schwarz
$\frac1n\sum_j\log(2+|\zeta_j|)\le\log(2+\sqrt{T_n})\le 1+\sqrt{T_n}$. Collecting
terms and applying Cauchy--Schwarz in $\omega$,
\[
        \text{LHS}\le \frac{4\,\bE T_n}{R}
        +\frac{C\,\bE\bigl[T_n(1+\sqrt{T_n})\bigr]}{R^2}
        \le \frac{C'(\sqrt K+K^{3/4})}{R},
\]
using $\bE T_n\le\sqrt K$ and $\bE T_n^{3/2}\le (\bE T_n^2)^{3/4}\le K^{3/4}$.
\end{proof}

\subsection{The one-sided bound is unconditional}
Truncating a logarithm from below can only increase it. This trivial remark gives
the upper bound on $\Lambda_n$ for free, and it is the half of the argument that
does not require any local information about eigenvalue repulsion.

\begin{lemma}[Upper bound]\label{lem:upper}
Assume $L_n\to\mu$ weakly in probability for a compactly supported probability
measure $\mu$, and assume the conclusion of Lemma~\ref{lem:tail}. Then
\[
        \limsup_{n\to\infty}\ \Lambda_n\ \le\ I(\mu)
        \qquad\text{in probability},
\]
i.e.\ for every $\eta>0$, $\bP\bigl(\Lambda_n>I(\mu)+\eta\bigr)\to0$. The same
statement holds for $\bE\Lambda_n$.
\end{lemma}

\begin{proof}
Since
\[
 \log|z-w|\le\Phi_{\eps,R}(z,w)+\log^+\frac{|z-w|}{R},
 \qquad
 \log^+\frac{|z-w|}R
 \le\log(2+|z|+|w|)\mathbbm 1_{\{\max(|z|,|w|)>R/2\}},
\]
Lemma~\ref{lem:tail} gives
\[
        \Lambda_n\le\iint\Phi_{\eps,R}\,dM_n+\mathcal T_n(R),
        \qquad \bE\mathcal T_n(R)\le C_K/R .
\]
Next, directly from the definitions,
\begin{equation}\label{eq:diagcorr}
        \iint\Phi_{\eps,R}\,dM_n
        =\frac{n}{n-1}\iint\Phi_{\eps,R}\,dL_n\,dL_n-\frac{\log\eps}{n-1},
\end{equation}
since $\Phi_{\eps,R}(z,z)=\log\eps$. As $\Phi_{\eps,R}$ is bounded and continuous,
weak convergence of $L_n$ in probability yields
$\iint\Phi_{\eps,R}dM_n\to\iint\Phi_{\eps,R}\,d\mu\,d\mu$ in probability. Now fix
$R$ larger than the diameter of $\supp\mu$ and let $n\to\infty$, then
$R\to\infty$; finally $\Phi_{\eps,R}=\log_\eps|z-w|\downarrow\log|z-w|$ as
$\eps\downarrow0$, so monotone convergence gives
$\iint\Phi_{\eps,R}d\mu\,d\mu\downarrow I(\mu)>-\infty$, which is the claim.
\end{proof}

\subsection{From a one-sided bound plus a mean to \texorpdfstring{$L^1$}{L1}}
The following elementary lemma is what allows us to work with expectations only. It
says that for a sequence dominated from above by something that converges, the
convergence of the mean forces convergence in $L^1$; no concentration inequality is
needed.

\begin{lemma}[$L^1$ upgrade]\label{lem:usc_upgrade}
Let $(X_n)$ be integrable real random variables and suppose that for each $k$ there
are random variables $X_n^{(k)}$ and $\mathcal T_n^{(k)}\ge0$ with
\begin{enumerate}
\item[(a)] $X_n\le X_n^{(k)}+\mathcal T_n^{(k)}$ and
$\sup_n\|X_n^{(k)}\|_\infty<\infty$ for each $k$;
\item[(b)] $X_n^{(k)}\to c_k$ in probability as $n\to\infty$, and
$\limsup_n\bE\mathcal T_n^{(k)}\le\rho_k$;
\item[(c)] $c_k\to L$ and $\rho_k\to0$ as $k\to\infty$;
\item[(d)] $\bE X_n\to L$.
\end{enumerate}
Then $X_n\to L$ in $L^1$, hence in probability.
\end{lemma}

\begin{proof}
Fix $k$. By (a), $(X_n-X_n^{(k)})^+\le\mathcal T_n^{(k)}$, hence
\[
        \bE|X_n-X_n^{(k)}|
        =\bE(X_n^{(k)}-X_n)+2\bE(X_n-X_n^{(k)})^+
        \le \bE X_n^{(k)}-\bE X_n+2\bE\mathcal T_n^{(k)} .
\]
By (a) and (b) and bounded convergence, $\bE X_n^{(k)}\to c_k$; with (d),
$\limsup_n\bE|X_n-X_n^{(k)}|\le c_k-L+2\rho_k$. Also
$\bE|X_n^{(k)}-c_k|\to0$. Therefore
$\limsup_n\bE|X_n-L|\le (c_k-L)+2\rho_k+|c_k-L|$, and letting $k\to\infty$
finishes the proof.
\end{proof}

\subsection{A Christoffel bound replacing level repulsion}
Finally we treat the near-diagonal contribution, which is the only genuinely local
input. For determinantal Coulomb gases it is controlled by the crudest possible
bound $\rho_2\le\rho_1\otimes\rho_1$ together with a uniform bound on the one-point
function. No eigenvalue repulsion is needed --- which is fortunate, since it is
exactly the repulsion estimates that are hard to obtain uniformly in a parameter.

Consider a $\beta=2$ two-dimensional Coulomb gas: $n$ points in $\bC$ with joint
density
\begin{equation}\label{eq:gas}
        \frac1{Z_n}\prod_{i<j}|z_i-z_j|^2\,
        e^{-n\sum_k V(z_k)},
\end{equation}
with $V\in C^2(\bC)$ growing fast enough. Such a gas is determinantal with
correlation kernel $K_n(z,w)=\sum_{k<n}p_k(z)\overline{p_k(w)}e^{-\frac n2(V(z)+V(w))}$,
where $(p_k)$ are the orthonormal polynomials for the weight $e^{-nV}dA$; the
one- and two-point functions are $\rho_1=K_n(z,z)$ and
$\rho_2(z,w)=K_n(z,z)K_n(w,w)-|K_n(z,w)|^2$, and $K_n$ is the kernel of an
orthogonal projection of rank $n$.

\begin{lemma}[Christoffel bound]\label{lem:christoffel}
If $\sup_{\bC}\Delta V\le M<\infty$ then
$K_n(z,z)\le \frac n\pi\,e^{M/4}$ for all $z\in\bC$ and all $n\ge1$.
\end{lemma}

\begin{proof}
By the reproducing property, $K_n(z,z)=\sup\{|f(z)|^2:f\in\mathcal H_n,\
\|f\|_{L^2(dA)}=1\}$, where $\mathcal H_n$ is the span of the $p_ke^{-nV/2}$, i.e.\
the space of $f=pe^{-nV/2}$ with $p$ a polynomial of degree $<n$. Fix such an $f$
with $\|f\|=1$ and fix $z$. The function
\[
        u(w):=\log|f(w)|^2+\frac{nM}4|w-z|^2
\]
satisfies, in the sense of distributions,
$\Delta u=4\pi\textstyle\sum_{p(\zeta)=0}\delta_\zeta-n\Delta V+nM\ge0$, so $u$ is
subharmonic, hence so is $e^{u}$. The sub-mean-value property on the disc
$D(z,r)$ with $r=n^{-1/2}$ gives
\[
        |f(z)|^2=e^{u(z)}\le\frac1{\pi r^2}\int_{D(z,r)}e^{u}\,dA
        =\frac1{\pi r^2}\int_{D(z,r)}|f(w)|^2e^{\frac{nM}4|w-z|^2}dA(w)
        \le\frac{e^{M/4}}{\pi r^2}\|f\|^2 ,
\]
since $\frac{nM}4|w-z|^2\le \frac{M}4$ on $D(z,r)$. As $r^{-2}=n$, the claim
follows.
\end{proof}

\begin{proposition}[Near-diagonal estimate]\label{prop:neardiag}
For a gas \eqref{eq:gas} with $\sup\Delta V\le M$ one has, for all $0<\eps\le1$ and
all $n\ge2$,
\[
        \frac1N\,\bE\sum_{i\ne j}\mathbbm 1_{\{|\zeta_i-\zeta_j|\le\eps\}}
        \bigl|\log|\zeta_i-\zeta_j|\bigr|
        \;\le\;
        \frac{n}{n-1}\,e^{M/4}\,\eps^2\Bigl(|\log\eps|+\tfrac12\Bigr).
\]
In particular the bound is uniform in $n$ and tends to $0$ as $\eps\downarrow0$.
\end{proposition}

\begin{proof}
The left-hand side equals
$\frac1N\iint_{|z-w|\le\eps}\bigl|\log|z-w|\bigr|\rho_2(z,w)\,dA(z)dA(w)$. Since
$\rho_2\le K_n(z,z)K_n(w,w)$, Lemma~\ref{lem:christoffel} and
$\int K_n(z,z)dA=n$ give the bound
\begin{align*}
 &\frac1N\Bigl(\sup_wK_n(w,w)\Bigr)\int K_n(z,z)
 \Bigl(\int_{|w-z|\le\eps}\bigl|\log|z-w|\bigr|dA(w)\Bigr)dA(z)\\
 &\qquad\le \frac{1}{N}\frac{n e^{M/4}}{\pi}\,n\,
 \pi\eps^2\bigl(|\log\eps|+\tfrac12\bigr).
\end{align*}
using $\int_{|w-z|\le\eps}|\log|z-w||dA(w)=2\pi\int_0^\eps r|\log r|dr
=\pi\eps^2(|\log\eps|+\tfrac12)$ for $0<\eps\le1$.
\end{proof}

\begin{proposition}[Energy convergence for determinantal gases]\label{prop:gasenergy}
Let $(\zeta^{(n)})$ be gases \eqref{eq:gas} with potentials $V_n$ satisfying
$\sup_n\sup_\bC\Delta V_n\le M$, with $L_n\to\mu$ weakly in probability for a
compactly supported $\mu$ with $I(\mu)>-\infty$, and with
$\sup_n\bE T_n^2<\infty$. Then $\bE\Lambda_n\to I(\mu)$ and $\Lambda_n\to I(\mu)$
in $L^1$. Moreover $\sup_n\bE|\Lambda_n|<\infty$, with a bound depending only on
$M$, on $\sup_n\bE T_n^2$ and on $\diam\supp\mu$.
\end{proposition}

\begin{proof}
Write $\Lambda_n=\iint\Phi_{\eps,R}dM_n+E^{\rm near}_n(\eps)+E^{\rm far}_n(R)$
where $E^{\rm near}$ collects the part of the integrand supported on
$\{|z-w|\le\eps\}$ and $E^{\rm far}$ that on $\{|z-w|>R\}$. By
Proposition~\ref{prop:neardiag}, $\bE|E^{\rm near}_n(\eps)|\le
2e^{M/4}\eps^2(|\log\eps|+\tfrac12)$ uniformly in $n$; by Lemma~\ref{lem:tail},
$\bE|E^{\rm far}_n(R)|\le C/R$ uniformly in $n$. For the truncated term, use
$\rho_2=\rho_1\otimes\rho_1-|K_n|^2$: since $K_n$ is a rank-$n$ projection,
$\iint|K_n(z,w)|^2dA\,dA=\operatorname{tr}K_n=n$, so
\[
        \Bigl|\;\frac1N\iint\Phi_{\eps,R}\,|K_n|^2\Bigr|
        \le\frac{\max(|\log\eps|,\log R)\,n}{N}\longrightarrow0 ,
\]
while $\frac1N\iint\Phi_{\eps,R}\rho_1\otimes\rho_1
=\frac{n^2}{N}\iint\Phi_{\eps,R}\,d\bE L_n\,d\bE L_n
\to\iint\Phi_{\eps,R}\,d\mu\,d\mu$ by weak convergence. Letting $n\to\infty$,
then $R\to\infty$ and $\eps\downarrow0$, gives $\bE\Lambda_n\to I(\mu)$. The $L^1$
statement follows from Lemma~\ref{lem:usc_upgrade} with
$X_n=\Lambda_n$, $X^{(k)}_n=\iint\Phi_{\eps_k,R_k}dM_n$ and
$\mathcal T^{(k)}_n$ the far-field error, using Lemma~\ref{lem:upper}.

For the last assertion, split $\bE|\Lambda_n|$ according to $|z-w|\le1$,
$1<|z-w|\le R$ and $|z-w|>R$. The first part is at most $e^{M/4}$ by
Proposition~\ref{prop:neardiag} with $\eps=1$; the second is at most $\log R$; the
third is at most $C_K/R$ by Lemma~\ref{lem:tail}. Taking $R$ fixed gives a
bound independent of $n$.
\end{proof}

\subsection{The general criterion}
We can now state the criterion that organizes all the limit theorems below.

\begin{theorem}[Criterion]\label{th:criterion}
Let $(A_n,B_n)$ be nondegenerate random pencils with entries of variance $1/n$, and
let $\Omega\subseteq\bC$ be open with $\si(\bCP^1\setminus\Omega)=0$. Assume:
\begin{enumerate}
\item[\rm (G)] for almost every $\la\in\Omega$, the empirical spectral distribution
of $\widehat C_n(\la)$ converges weakly in probability to a compactly supported
probability measure $\mu_\la$ with $I(\mu_\la)>-\infty$;
\item[\rm (T)] $\sup_n\bE T_n(\la)^2<\infty$, locally uniformly on $\Omega$;
\item[\rm (E)] for almost every $\la\in\Omega$,
$\bE\Lambda_n(\la)\to I(\mu_\la)$;
\item[\rm (B)] $\sup_n\sup_{\la\in K}\bE|\Lambda_n(\la)|<\infty$ for every compact
$K\subset\Omega$.
\end{enumerate}
Set $G(\la):=I(\mu_\la)$ almost everywhere and
$\Psi:=\si+\frac1{2\pi}\Delta_\la G$ on $\Omega$. Then
$\Lambda_n(\la)\to G(\la)$ in $L^1(\bP)$ for almost every $\la\in\Omega$, and every
weak limit point $\nu$ of $\mu_n$ satisfies $\nu|_\Omega=\Psi|_\Omega$. If moreover
$\Psi(\Omega)=1$, then $\mu_n\to\Psi$ weakly in probability on $\bCP^1$. If $G$
is a function of $Y$ alone, the limit is described by the distributional master
formula of Proposition~\ref{prop:master}.
\end{theorem}

\begin{proof}
For almost every $\la\in\Omega$, hypotheses (G), (T), (E) and
Lemmas~\ref{lem:tail}, \ref{lem:upper}, \ref{lem:usc_upgrade} give
$\Lambda_n(\la)\to G(\la)$ in $L^1(\bP)$, exactly as in the proof of
Proposition~\ref{prop:gasenergy}. By (B), the functions
$\bE|\Lambda_n(\la)-G(\la)|$ are locally dominated, so dominated convergence
and Fubini give
$\bE\int_K|\Lambda_n-G|\,dx\,dy\to0$ for every compact $K\subset\Omega$, i.e.\
$\Lambda_n\to G$ in $L^1_{\rm loc}(\Omega)$ in probability.

Let now $f\in C_c^\infty(\Omega)$. By Lemma~\ref{lem:pot_basic}(iii) and the key
identity \eqref{eq:key},
\[
        \int f\,d\mu_n=\int f\,d\si+\frac1{2\pi}\int(\Delta f)\,P^{\mu_n}\,dx\,dy
        =\int f\,d\si+\frac1{2\pi}\int(\Delta f)\,\Lambda_n\,dx\,dy ,
\]
where the $\la$-independent constant of \eqref{eq:key} drops out because
$\int_{\bC}\Delta f\,dx\,dy=0$. Hence
$\int f\,d\mu_n\to\int f\,d\si+\frac1{2\pi}\int(\Delta f)G=\int f\,d\Psi$ in
probability. Thus any weak limit point of $\mu_n$ agrees with $\Psi$ on $\Omega$. If $\Psi(\Omega)=1$ then
Lemma~\ref{lem:mass} upgrades this to weak convergence on $\bCP^1$. The last claim
is Proposition~\ref{prop:master}.
\end{proof}

\begin{remark}\label{rem:comparison}
It is worth comparing the hypotheses of Theorem~\ref{th:criterion} with those used
previously. The uniform circular law (UCL) --- convergence of the spectral
distribution uniformly in $\la$ on compacta --- is not assumed: almost-everywhere
pointwise convergence is enough, since the uniformity required for the
potential-theoretic step is supplied by (B) and dominated convergence. The
logarithmic tail hypothesis (LT) has become Lemma~\ref{lem:tail}. The small
repulsion hypothesis (SR), a statement in probability, uniform in $\la$, about
near-diagonal pairs, has been replaced by (E), a statement about the convergence of
a single scalar expectation at almost every fixed $\la$; and (E) is a theorem, not a
hypothesis, whenever the relevant ensemble is a determinantal Coulomb gas with
bounded $\Delta V$ (Proposition~\ref{prop:gasenergy}). Finally, no hypothesis
excluding concentration near $\bRP^1$ appears: Lemma~\ref{lem:mass} makes it
superfluous.
\end{remark}

\section{The Hermitian case: an unconditional theorem}\label{sec:hermitian}

Recall that the $GUE_n$-ensemble is the probability distribution on the space
$\HH_n$ of Hermitian $n\times n$ matrices with density proportional to
$e^{-\frac n2\operatorname{tr}H^2}$, so that
$\bE H_{ii}^2=\bE|H_{ij}|^2=1/n$ and the spectrum fills $[-2,2]$ in the limit. As
in \cite{ShZa2} we record the exactly solvable case $n=2$.

\begin{THEOQ}[{\cite[Thm.~3]{ShZa2}}]\label{th:GUE2}
If $A,B$ are independent $GUE_2$ matrices, the distribution of level crossings is
\begin{equation}\label{eq:densGUE2}
        \cP_{GUE_2}(x,y)\,dx\,dy=\frac4\pi\,\frac{|y|\,dx\,dy}{(1+x^2+y^2)^3},
\end{equation}
i.e.\ it has density $2|Y|$ with respect to $\si$.
\end{THEOQ}

The key structural point is that the normalized Hermitian pencil is exactly an
elliptic Ginibre matrix.

\begin{lemma}[Exact identification]\label{lem:identification}
Let $A_n,B_n$ be independent $GUE_n$ matrices and $\la\in\bC$. Then
$\widehat C_n(\la)=(A_n+\la B_n)/\sqrt{1+|\la|^2}$ is exactly a complex elliptic
Ginibre matrix with parameter $\tau(\la)=(1+\la^2)/(1+|\la|^2)$; that is,
$e^{-i\theta}\widehat C_n(\la)$, where $\tau=|\tau|e^{2i\theta}$, has the
distribution of
$\sqrt{\tfrac{1+|\tau|}2}\,H_1+i\sqrt{\tfrac{1-|\tau|}2}\,H_2$ with $H_1,H_2$
independent $GUE_n$. Moreover $|\tau(\la)|^2=1-Y(\la)^2$, so
$|\tau(\la)|=1$ exactly on $\bRP^1=\bR\cup\{\infty\}$.
\end{lemma}

\begin{proof}
The entries of $\widehat C_n(\la)$ are jointly centered complex Gaussian, so their
law is determined by the covariance and pseudo-covariance. A direct computation
gives, for all $i,j$,
\[
        \bE|\widehat c_{ij}|^2=\frac1n,\qquad
        \bE\bigl[\widehat c_{ij}\widehat c_{ji}\bigr]
        =\frac{\bE|a_{ij}|^2+\la^2\,\bE|b_{ij}|^2}{1+|\la|^2}=\frac{\tau(\la)}n,
        \qquad \bE\bigl[\widehat c_{ij}^{\,2}\bigr]=0\ (i\ne j),
\]
all other second moments vanishing, which is precisely the covariance structure of
the elliptic Ginibre ensemble with parameter $\tau$; multiplying by the phase
$e^{-i\theta}$ makes the parameter real and rotates the spectrum. For the last
statement, writing $\la=x+iy$,
\[
        |\tau|^2=\frac{|1+\la^2|^2}{(1+|\la|^2)^2}
        =\frac{(1+x^2-y^2)^2+4x^2y^2}{(1+x^2+y^2)^2}
        =1-\frac{4y^2}{(1+x^2+y^2)^2}=1-Y^2 . \qedhere
\]
\end{proof}

The complex elliptic Ginibre ensemble with parameter $\tau\in[0,1)$ has joint
eigenvalue density of the form \eqref{eq:gas} with
\begin{equation}\label{eq:Vtau}
        V_\tau(z)=\frac{|z|^2-\tau\,\Re(z^2)}{1-\tau^2},
        \qquad
        \Delta V_\tau=\frac4{1-\tau^2},
\end{equation}
see \cite{SCS,FKS,FeinbergZee,ForresterBook}, and its empirical spectral
distribution converges to $\Ell_\tau$. Note that $\Ell_\tau$ is exactly the
equilibrium measure of $V_\tau$: its density $\frac1{4\pi}\Delta V_\tau=
\frac1{\pi(1-\tau^2)}$ on an ellipse of area $\pi(1-\tau^2)$ has total mass one.

\begin{proof}[Proof of Theorem~\ref{th:GUEmain}]
Let $\Omega:=\bC\setminus\bR$, so $\si(\bCP^1\setminus\Omega)=\si(\bRP^1)=0$. We
verify the hypotheses of Theorem~\ref{th:criterion}.

Fix a compact $K\subset\Omega$ and let $\delta>0$ be such that
$|\tau(\la)|\le1-\delta$ for $\la\in K$; this is possible since $|\tau|<1$ off
$\bR$ and $|\tau|$ is continuous. By Lemma~\ref{lem:identification} and
\eqref{eq:Vtau}, for $\la\in K$ the eigenvalues of $\widehat C_n(\la)$ form, after
a rotation, a determinantal gas \eqref{eq:gas} with
$\sup\Delta V_{|\tau|}=\frac4{1-|\tau|^2}\le\frac4{\delta}=:M_\delta$, uniformly on
$K$.

(G) holds with $\mu_\la=\Ell_{\tau(\la)}$: this is the elliptic law for Gaussian
elliptic ensembles, \cite{SCS,FKS,ForresterBook}. Its energy is finite, indeed
$I(\Ell_\tau)=-\tfrac14$ by Theorem~\ref{th:energy}.

(T) holds since $\bE\|\widehat C_n\|_{HS}^2=n$ and $T_n$ is a normalized sum of
squares of Gaussians, so $\bE T_n^2=O(1)$.

(E) and (B) hold by Proposition~\ref{prop:gasenergy}, applied with the uniform
bound $M_\delta$ and with $\supp\Ell_\tau$ contained in the disc of radius $2$:
$\bE\Lambda_n(\la)\to I(\Ell_{\tau(\la)})=-\tfrac14$, and
$\sup_n\sup_{\la\in K}\bE|\Lambda_n(\la)|<\infty$ by the last assertion of that
proposition, all bounds depending on $K$ only through $\delta$.

Hence $G(\la)=I(\Ell_{\tau(\la)})\equiv-\tfrac14$ is \emph{constant} on $\Omega$,
so $\Delta_\la G=0$ and $\Psi=\si$ on $\Omega$. Since $\si(\Omega)=1$,
Theorem~\ref{th:criterion} gives $\mu_n\to\si$ weakly in probability on $\bCP^1$.

The final statement is identical: if $A_n$ and $B_n$ are independent Gaussian
ensembles chosen so that $\widehat C_n(\la)$ is a complex elliptic Ginibre matrix
with some parameter $\tau(\la)$ of modulus $<1$ off a $\si$-null set, the same
proof applies verbatim, since only the values $I(\Ell_{\tau})=-\tfrac14$ enter. For
instance if $A_n$ is $GUE_n$ and $B_n$ is complex Ginibre then
$\bE[\widehat c_{ij}\widehat c_{ji}]=\frac1{n(1+|\la|^2)}$, so
$\tau(\la)=(1+|\la|^2)^{-1}\in(0,1]$ and $|\tau|<1$ off $\la=0$.
\end{proof}

\begin{corollary}\label{cor:conj}
The limit $\Psi=\lim_n\cP^{GUE}_n$ conjectured in \cite{ShZa2} exists and equals
the uniform measure on $\bCP^1$. In particular the Wigner-type and Girko-type
level-crossing laws coincide.
\end{corollary}

\begin{remark}[Why the numerics of \cite{ShZa2} suggested otherwise]\label{rem:rate}
The finite-$n$ Hermitian laws are genuinely non-uniform: $GUE_2$ has density
$2|Y|$.  Theorem~\ref{th:GUEmain} proves convergence but does not give a rate.
The key identity shows that the deviation of $\bE\mu_n$ from $\si$ is governed by
the oscillation of $\la\mapsto\bE\Lambda_n(\la)$, i.e.\ by the finite-$n$
correction to the logarithmic energy of the elliptic ensemble.  Coulomb-gas
free-energy heuristics suggest a correction of order $n^{-1}\log n$, and the
simulations in \S\ref{sec:numerics} are consistent with that scale, but we do not
prove it here.  The endpoint oscillation decreases from the exact value
$\tfrac12(1-\log2)=0.153426\ldots$ at $n=2$ to about $0.023$ at $n=64$, while
$\bE_{\mu_n}|Y|$ decreases from $2/3$ to about $0.535$.  Extrapolating only from
$n\le6$, as in Fig.~3 of \cite{ShZa2}, therefore naturally suggests a non-uniform
limit even though the true limit is uniform.
\end{remark}

\begin{remark}[Consistency at $n=2$]
Example~\ref{ex:GUE2} predicts, from the exact $GUE_2$ density
\eqref{eq:densGUE2}, that the $2\times2$ complex elliptic Ginibre ensemble has
energy profile satisfying $G_2(1)-G_2(0)=-\tfrac12(1-\log2)=-0.153426\ldots$.
Direct simulation of the $2\times2$ ensemble, using the closed form
$\zeta_1-\zeta_2=\sqrt{(\widehat c_{11}-\widehat c_{22})^2+4\widehat c_{12}
\widehat c_{21}}$ and $1.2\cdot10^8$ samples, gives
$-0.153371\pm0.000043$. This tests the key identity \eqref{eq:key}, the master
formula \eqref{eq:master} and the coordinate $Y$ simultaneously.
\end{remark}

\section{The real symmetry classes}\label{sec:real}

\subsection{What changes}
For real ensembles two things are different. First, the normalized pencil is no
longer determinantal, so Proposition~\ref{prop:gasenergy} does not apply and
hypothesis (E) of Theorem~\ref{th:criterion} must be imposed. Second, crossings can
be real, so $\bRP^1$ must be dealt with; by Lemma~\ref{lem:mass} this costs
nothing.

\subsection{Real i.i.d.\ pencils}
Here the situation is cleaner than one might expect: the correlation between
transposed off-diagonal entries vanishes identically, so the limiting spectral law
is circular even though an individual entry may have nonzero pseudo-variance.

\begin{lemma}\label{lem:iid_tau}
Let $A_n,B_n$ be independent random matrices with i.i.d.\ real (or complex) entries
of mean zero and variance $1/n$. Then, for every $\la\in\bC$, the entries of
$\widehat C_n(\la)$ are i.i.d.\ with mean zero and $\bE|\widehat c_{ij}|^2=1/n$,
and $\bE[\widehat c_{ij}\widehat c_{ji}]=0$ for $i\ne j$. Consequently the
empirical spectral distribution of $\widehat C_n(\la)$ converges to the circular
law $\Ell_0$ for every $\la$.
\end{lemma}

\begin{proof}
For $i\ne j$ the pairs $(a_{ij},b_{ij})$ and $(a_{ji},b_{ji})$ are independent,
hence $\bE[\widehat c_{ij}\widehat c_{ji}]=\bE\widehat c_{ij}\,\bE\widehat c_{ji}=0$.
The entries $\widehat c_{ij}$ are i.i.d.\ complex with $\bE|\widehat c_{ij}|^2=1/n$.
Their pseudo-variance $\bE\widehat c_{ij}^{\,2}$ is generally nonzero and depends
on $\la$, but the circular law of Girko and Tao--Vu \cite{Gi2,TV} requires only
mean zero and finite variance and is insensitive to this scalar pseudo-variance.
\end{proof}

\begin{theorem}\label{th:iid}
Let $A_n,B_n$ be independent $n\times n$ random matrices with i.i.d.\ real or
complex entries of mean zero, variance $1/n$ and finite fourth moment. Assume (E), i.e.\ that for almost every $\la$,
\[
        \bE\,\Lambda_n(\la)\longrightarrow -\tfrac14 ,
\]
together with the local uniform bound (B). Then the empirical level-crossing
measures converge weakly in probability to the uniform measure $\si$ on $\bCP^1$.
\end{theorem}

\begin{proof}
Hypotheses (G) and (T) of Theorem~\ref{th:criterion} hold by
Lemma~\ref{lem:iid_tau} and the fourth-moment assumption, with
$\mu_\la=\Ell_0$ for every $\la$, so $G\equiv I(\Ell_0)=-\tfrac14$ is constant and
$\Psi=\si$ on $\Omega=\bC$. Apply Theorem~\ref{th:criterion}.
\end{proof}

This is the promised strengthening of the conditional Girko-type statement: the
global spectral input is now a theorem rather than a hypothesis, uniformity in
$\la$ has been eliminated, the logarithmic tail hypothesis has become
Lemma~\ref{lem:tail}, and the hypothesis on $\bRP^1$ has disappeared. What remains,
(E), concerns a single scalar and is precisely the statement that the near-diagonal
pairs of eigenvalues do not contribute to the leading order of the logarithmic
energy. It follows from any two-point Wegner-type estimate: if
\begin{equation}\label{eq:wegner}
        \bE\,\#\bigl\{(i,j):i\ne j,\ |\zeta_i|,|\zeta_j|\le R,\
        |\zeta_i-\zeta_j|\le r\bigr\}\le C_R\,n^2r^2
        \quad\text{uniformly in }n,\ 0<r<1,
\end{equation}
then the dyadic decomposition of $\mathbbm1_{\{|z-w|\le\eps\}}|\log|z-w||$ gives
exactly the bound of Proposition~\ref{prop:neardiag}, and (E) follows. Estimate
\eqref{eq:wegner} is what local universality for non-Hermitian matrices
\cite{TVnonHerm,BYY} predicts, but we are not aware of a reference establishing it
in this generality.

\subsection{Real symmetric pencils}
Let $GOE_n$ be the Gaussian orthogonal ensemble on real symmetric matrices, with
off-diagonal entries of variance $1/n$ and diagonal entries of variance $2/n$. The
relevant exact result and conjecture from \cite{ShZa2} are:

\begin{THEOQ}[{\cite[Thm.~5]{ShZa2}}]\label{th:GOE2}
If $A,B$ are independent $GOE_2$ matrices, the level crossings are exactly uniform
on $\bCP^1$.
\end{THEOQ}

\begin{conjecture}[{\cite{ShZa2}}]\label{conj:GOE}
The conclusion of Theorem~\ref{th:GOE2} holds for every $n\ge3$.
\end{conjecture}

Here the normalized pencil is a \emph{complex symmetric} Gaussian matrix:
$\widehat c_{ij}=\widehat c_{ji}$ with $\bE|\widehat c_{ij}|^2=1/n$ and
$\bE\widehat c_{ij}^{\,2}=\tau(\la)/n$.  This is a singular mirror-correlation
class and is not the determinantal elliptic Ginibre ensemble used in
\S\ref{sec:hermitian}.  To keep the reduction logically independent of a
particular global-law theorem, we isolate the required spectral input:
\begin{equation}\label{eq:Gsym}
 \tag{$G_{\rm sym}$}
 L_n(\la)\Longrightarrow \Ell_{\tau(\la)}
 \quad\text{in probability for a.e. }\la\in\bC\setminus\bR .
\end{equation}
Here a phase rotation is understood when $\tau(\la)$ is not real.  At the circular
point $\tau=0$, recent finite-$n$ analyses of the Gaussian complex-symmetric class
recover the circular law in the large-$n$ limit \cite{AFS,CM}.  The anisotropic
family required in \eqref{eq:Gsym}, which connects that point to the real-symmetric
endpoint, is a different spectral input and is kept explicit here.

\begin{theorem}\label{th:GOEcond}
Let $A_n,B_n$ be independent $GOE_n$ matrices.  Assume the global spectral input
\eqref{eq:Gsym}, the scalar energy condition
\[
 \bE\Lambda_n(\la)\longrightarrow I(\Ell_{\tau(\la)})=-\tfrac14
 \quad\text{for a.e. }\la\notin\bR,
\]
and the local uniform bound (B) on $\bC\setminus\bR$.  Then
$\mu_n\to\si$ weakly in probability on $\bCP^1$; in particular the asymptotic
form of Conjecture~\ref{conj:GOE} follows from these two spectral inputs.
\end{theorem}

\begin{proof}
Hypothesis (T) is immediate from Gaussian moments.  Under \eqref{eq:Gsym}, (G) of
Theorem~\ref{th:criterion} holds with $\mu_\la=\Ell_{\tau(\la)}$; the assumed
energy convergence is (E), and (B) is part of the hypotheses.  By
Theorem~\ref{th:energy}, $G(\la)=I(\Ell_{\tau(\la)})\equiv-\tfrac14$ off
$\bRP^1$.  Theorem~\ref{th:criterion}, followed by Lemma~\ref{lem:mass}, gives the
claim.
\end{proof}

\begin{remark}[Reinterpretation of the pseudo-covariance conjecture]
Previous versions of this paper reduced Conjecture~\ref{conj:GOE} to the assertion that the
pseudo-covariance $\tau(\la)$ ``does not affect the leading $n^2$ term of the
log-discriminant''. Theorem~\ref{th:energy} explains the energy mechanism behind
that expectation: $\tau$ can deform the macroscopic spectrum from a disc towards a
segment while leaving the logarithmic energy of the elliptic law unchanged.  In
the complex-symmetric class, however, two logically separate inputs remain:
establish the global law \eqref{eq:Gsym}, and justify the interchange of limits in
$\Lambda_n\to I(\lim L_n)$ by excluding an anomalous contribution from very close
eigenvalue pairs.
\end{remark}

\begin{remark}[Numerical status]
The simulations of \S\ref{sec:numerics} indicate that for $GOE$ pencils the
profile $\bE\Lambda_n(\la)$ is flat, within the Monte Carlo resolution, for every
$n$ tested; see the middle panel of Figure~\ref{fig:energy}.  This is consistent
with, and would be stronger than, Conjecture~\ref{conj:GOE} if it held exactly.
It suggests that an additional finite-$n$ mechanism may be present, but we do not
know such a mechanism.
\end{remark}

\section{Diagonal deformation: exact endpoint and a conjectural crossover}\label{sec:transition}

It was observed in \cite{ShZa2} that enlarging the main diagonal of a real
symmetric pencil produces a strong departure from the spherical uniform law.  The
following discussion separates what can be proved directly from what is presently
supported only by numerics.  Write
\begin{equation}\label{eq:scaled}
        A_n^{(s)}=A_n^{\rm off}+s\,A_n^{\rm diag},
\end{equation}
where $A_n^{\rm off}$ and $A_n^{\rm diag}$ are the off-diagonal and diagonal parts
of a $GOE_n$ matrix, and similarly for $B_n^{(s)}$.  It is convenient in the
critical scaling window to set
\[
        s=\varsigma\sqrt{n/2},
\]
so that the diagonal entries have variance $\varsigma^2$, whereas the
off-diagonal entries retain variance $1/n$.

Remark~\ref{rem:SO2} applies for every $s$: the mean crossing measure is a function
of $|Y|$ alone.  Thus the apparent non-radiality in the affine $\la$-plane is an
axisymmetric dependence on the spherical height $Y$.

\subsection{The matrix-norm scale}
The diagonal deformation first becomes visible in operator norm near the
$\sqrt{n/\log n}$ scale.  This elementary fact is useful for orientation, but it
must not be confused with a theorem about nonnormal spectra or level crossings.

\begin{proposition}\label{prop:diag_opnorm}
Let $A_n$ be an unscaled $GOE_n$ matrix and let $s=s_n\ge0$.  Then almost surely
\[
 \bigl\|A_n^{(s_n)}-A_n\bigr\|_{\rm op}
 =O\!\left(|s_n-1|\sqrt{\frac{\log n}{n}}\right).
\]
Consequently this perturbation tends to zero in operator norm whenever
$|s_n-1|=o(\sqrt{n/\log n})$, in particular for every fixed $s$.
\end{proposition}

\begin{proof}
The difference is diagonal with entries
$(s_n-1)\sqrt{2/n}\,g_i$, where the $g_i$ are standard real Gaussians.  Hence
\[
 \bigl\|A_n^{(s_n)}-A_n\bigr\|_{\rm op}
 =|s_n-1|\sqrt{\frac2n}\max_{1\le i\le n}|g_i|
 =O\!\left(|s_n-1|\sqrt{\frac{\log n}{n}}\right)
\]
almost surely, by the standard Gaussian maximum estimate.
\end{proof}

For Hermitian matrices such a bound immediately controls eigenvalues.  For the
complex matrices $A_n^{(s)}+\la B_n^{(s)}$ with $\Im\la\ne0$, however, the problem
is nonnormal; operator-norm closeness alone does not justify stability of the
empirical spectrum, the logarithmic energy, or the crossing divisor.  We therefore
do not infer a level-crossing theorem from Proposition~\ref{prop:diag_opnorm}.

\subsection{The exactly diagonal Gaussian model}
At the opposite endpoint the crossing law can be computed without asymptotics.
Let $\si_{\bRP^1}$ denote the rotation-invariant probability measure on the real
projective line; in the affine coordinate it is the standard Cauchy law
$d\si_{\bRP^1}(x)=\pi^{-1}(1+x^2)^{-1}dx$.

\begin{theorem}[Diagonal endpoint]\label{th:diag_endpoint}
Let
\[
 A_n=\operatorname{diag}(a_1,\ldots,a_n),\qquad
 B_n=\operatorname{diag}(b_1,\ldots,b_n),
\]
where $(a_i,b_i)$ are i.i.d.\ standard real Gaussian vectors.  Then, for every
$n\ge2$,
\[
        \bE\mu_n=\si_{\bRP^1},
\]
and $\mu_n\Rightarrow\si_{\bRP^1}$ almost surely as $n\to\infty$.  For every
$n$, the mean logarithmic-energy profile $g_n(Y)=\bE\Lambda_n(Y)$ is, up to an
additive constant independent of $Y$,
\begin{equation}\label{eq:Gdiag}
        G_{\rm diag}(Y)=\frac12\log\frac{1+|Y|}{2}.
\end{equation}
Its distributional master formula is exactly the uniform measure on the equator
$\bRP^1\subset\bCP^1$.
\end{theorem}

\begin{proof}
For $i<j$ the two diagonal eigenvalues cross at
\[
        \la_{ij}=-\frac{a_i-a_j}{b_i-b_j}\in\bR.
\]
The discriminant contains each factor with exponent two, so
\[
        \mu_n=\binom n2^{-1}\sum_{i<j}\delta_{\la_{ij}}.
\]
The numerator and denominator in $\la_{ij}$ are independent centered Gaussians
with the same variance; hence $\la_{ij}$ has the standard Cauchy distribution.
This proves the finite-$n$ identity for $\bE\mu_n$.  The almost-sure convergence is
the strong law for $U$-statistics applied to bounded continuous test functions on
the compact space $\bRP^1$.

For the energy statement, each difference of two normalized diagonal entries has
the law of a fixed scalar multiple of an elliptic Gaussian variable.  Therefore
the expectation of every summand in $\Lambda_n$ is the same.  Lemma~\ref{lem:gaussenergy}
and $\sqrt{1-|\tau(\la)|^2}=|Y|$ give \eqref{eq:Gdiag}, modulo a
$Y$-independent constant, for every $n$.
For $Y\ne0$ one has
\[
 (1-Y^2)G_{\rm diag}'(Y)=\frac12\operatorname{sgn}(Y)(1-|Y|).
\]
Its jump at $Y=0$ is $1$, while its ordinary derivative away from $0$ is
$-1/2$.  Thus the absolutely continuous part in \eqref{eq:master} cancels, and the
unit jump contributes total mass one on the equator, exactly as claimed.
\end{proof}

\subsection{Critical scaling}
The exact diagonal model identifies the large-diagonal endpoint, but it does not
by itself prove what happens for a fixed positive $\varsigma$ in the scaling
$s=\varsigma\sqrt{n/2}$.  We record the interpretation suggested by
Figure~\ref{fig:diag} as a conjecture.

\begin{conjecture}[Diagonal crossover]\label{conj:transition}
For each $\varsigma\ge0$, the pencils with
$s=\varsigma\sqrt{n/2}$ satisfy
$\bE\mu_n\Rightarrow\Psi_\varsigma$, and their mean energy profiles
$g_n(Y)=\bE\Lambda_n(Y)$ converge to a deterministic profile $G_\varsigma(Y)$.
The family is axisymmetric, $\Psi_0=\si$, and for every $\varsigma>0$ the limiting
mean law is non-uniform.  Moreover
\[
        \Psi_\varsigma\Longrightarrow\si_{\bRP^1}
        \qquad (\varsigma\to\infty),
\]
and, modulo additive constants,
$G_\varsigma(Y)\to\tfrac12\log((1+|Y|)/2)$.  The deviation from uniformity is
monotone in $\varsigma$ in the sense suggested by the distribution functions in
Figure~\ref{fig:diag}.
\end{conjecture}

The numerical data show a smooth interpolation between the two endpoint profiles.
They do not constitute a proof of existence, non-uniformity for every positive
$\varsigma$, or monotonicity.  Establishing those statements requires a global law
and logarithmic-energy control for the diagonally deformed complex-symmetric
ensemble.

\section{Numerical experiments}\label{sec:numerics}

\subsection{Method}
Estimating $\bE\mu_n$ by locating the $n(n-1)$ roots of $\Delta_n$ is numerically
unstable beyond small $n$. The finite-$n$ identity of Remark~\ref{rem:finite_master} suggests a much better
route: it suffices to estimate the scalar function
\[
        g_n(Y):=\bE\,\Lambda_n(\la)
        =\frac1{n(n-1)}\,\bE\sum_{i\ne j}\log|\zeta_i(\la)-\zeta_j(\la)| ,
\]
which requires only eigenvalues, never roots of a high-degree polynomial. We
sample $\la=it$ with $t\in[0,1]$, so that $Y=2t/(1+t^2)$ sweeps $[0,1]$, and use
the same realizations of $(A_n,B_n)$ for all $t$ (common random numbers), which
makes the estimates of the \emph{differences} $g_n(Y)-g_n(0)$ --- the only thing
the answer depends on --- far more accurate than the individual values. The
distribution function of $|Y|$ is then recovered from \eqref{eq:cdf}, and
$\bE_{\mu_n}|Y|$ from the equivalent integrated form
\[
        \bE_{\mu_n}|Y|=\tfrac12+2g_n(0)-4\int_0^1Y\,g_n(Y)\,dY ,
\]
which involves no differentiation at all. As a check, for $n=2$ this reproduces
$\bE|Y|=2/3$, in agreement with Theorem~\ref{th:GUE2}.

\subsection{Results}
Figure~\ref{fig:energy} shows $g_n(Y)-g_n(0)$ for the three symmetry classes.
For $GUE$ pencils the profile is visibly nonconstant at small $n$ and flattens as
$n$ grows, in accordance with Theorem~\ref{th:GUEmain}; for $GOE$ pencils it is
flat within Monte Carlo resolution at every $n$ tested; for real i.i.d.\ pencils it
flattens from the other side. Figure~\ref{fig:rate} shows the observed decay: the
real i.i.d.\ oscillation is close to $n^{-1}$ over the simulated range, while the
$GUE$ oscillation decays more slowly and is consistent with the heuristic
$O(n^{-1}\log n)$ scale discussed in Remark~\ref{rem:rate}.  No convergence rate is
claimed from these fits.

\begin{figure}[ht]
\begin{center}
\includegraphics[width=\textwidth]{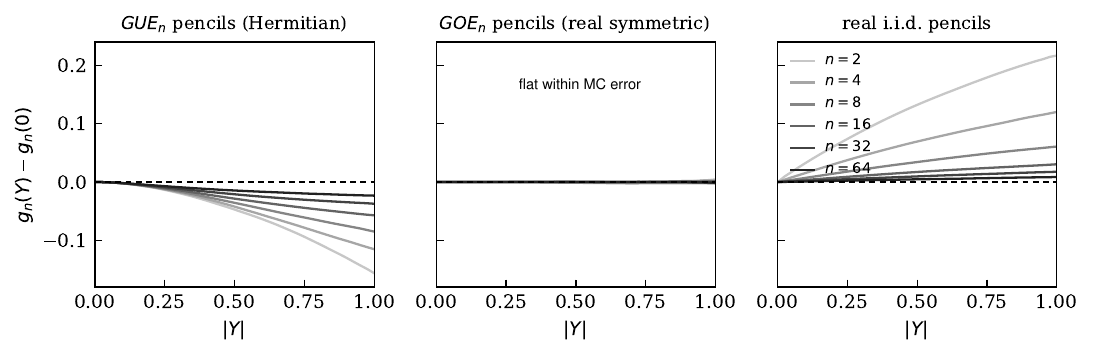}
\end{center}
\caption{The energy profile $g_n(Y)-g_n(0)$, whose constancy is equivalent to the
uniform law. Left: $GUE_n$ pencils, flattening as $n$ grows
(Theorem~\ref{th:GUEmain}). Middle: $GOE_n$ pencils, flat within Monte Carlo resolution at every tested $n$.
Right: real
i.i.d.\ pencils. Shades from light to dark correspond to $n=2,4,8,16,32,64$.}
\label{fig:energy}
\end{figure}

\begin{figure}[ht]
\begin{center}
\includegraphics[width=0.5\textwidth]{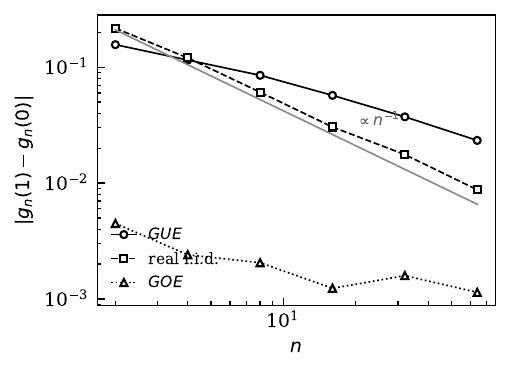}
\end{center}
\caption{Decay of the oscillation $|g_n(1)-g_n(0)|$ of the energy profile. For real
i.i.d.\ pencils the data are close to an $n^{-1}$ slope; for $GUE$ the decay is
slower; for $GOE$ the oscillation stays at the Monte Carlo noise floor.}
\label{fig:rate}
\end{figure}

Figure~\ref{fig:cdf} shows the distribution function of $|Y|$ under $\bE\mu_n$ for
$GUE$ pencils, computed from \eqref{eq:cdf}. At $n=2$ it agrees with the exact
answer $y^2$ of Theorem~\ref{th:GUE2}; as $n$ grows the curves migrate
monotonically towards the uniform law $y$. The mean $\bE_{\mu_n}|Y|$ takes the
values $0.667,0.635,0.605,0.577,0.553,0.535$ for $n=2,4,8,16,32,64$, against the
limit $1/2$.

\begin{figure}[ht]
\begin{center}
\includegraphics[width=0.52\textwidth]{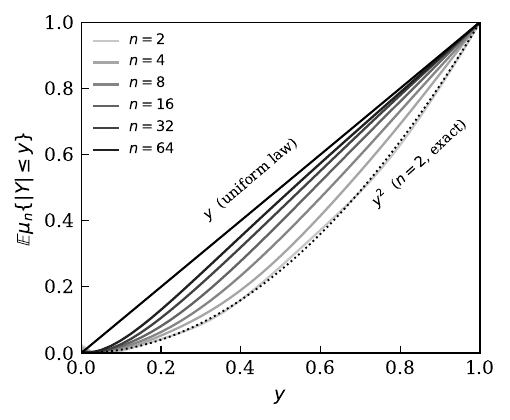}
\end{center}
\caption{Distribution function of $|Y|$ under $\bE\mu_n$ for $GUE_n$ pencils,
$n=2,4,8,16,32,64$ (light to dark), between the exact $n=2$ answer $y^2$ (dotted)
and the uniform law $y$ (solid).}
\label{fig:cdf}
\end{figure}

Figure~\ref{fig:diag} gives the numerical evidence behind
Conjecture~\ref{conj:transition}: as the diagonal scale $\varsigma$ increases, the
distribution of $|Y|$ moves from the uniform law towards concentration at $Y=0$,
and $\bE_{\mu_n}|Y|$ decreases from $1/2$ towards $0$.  At $n=48$ the measured
oscillations $g(1)-g(0)$ are $0.000,0.006,0.055,0.185,0.276$ for
$\varsigma=0,0.5,1,2,4$.  The exactly diagonal endpoint has oscillation
$\tfrac12\log2=0.3466\ldots$ by Theorem~\ref{th:diag_endpoint}.

\begin{figure}[ht]
\begin{center}
\includegraphics[width=\textwidth]{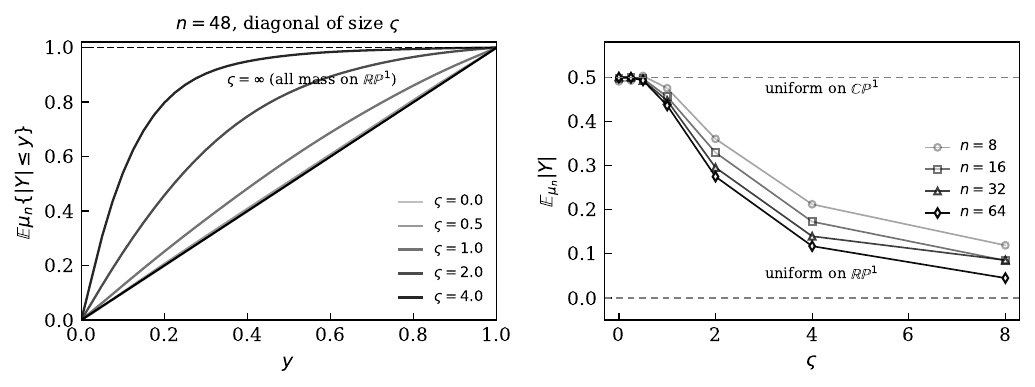}
\end{center}
\caption{Scaling the diagonal. Left: distribution function of $|Y|$ at $n=48$ for
diagonal scale $\varsigma=0,0.5,1,2,4$, with the two extremes --- the uniform law
on $\bCP^1$ (straight line) and the law concentrated on $\bRP^1$ (dashed) ---
indicated. Right: $\bE_{\mu_n}|Y|$ as a function of $\varsigma$ for
$n=8,16,32,64$, interpolating between $1/2$ and $0$.}
\label{fig:diag}
\end{figure}

We also verified Theorem~\ref{th:energy} directly by Monte Carlo integration over
the ellipse: the estimates of $I(\Ell_\tau)$ for
$\tau=0,0.2,\dots,0.9,0.99,1$ all agree with $-1/4$ to within $10^{-3}$.

\section{Open problems}\label{sec:problems}

\begin{enumerate}
\item Prove the scalar energy condition (E) for real i.i.d.\ pencils, for example
through a two-point estimate of the form \eqref{eq:wegner}.  This would make
Theorem~\ref{th:iid} unconditional under the stated moment assumptions.

\item For the $GOE$ pencil, establish the global complex-symmetric elliptic law
\eqref{eq:Gsym} in the precise form needed here and prove the corresponding energy
condition (E).  Separately, prove Conjecture~\ref{conj:GOE} in its exact finite-$n$
form.  The numerical flatness of $\bE\Lambda_n(\la)$ suggests that a hidden
finite-$n$ identity may exist.

\item Prove Conjecture~\ref{conj:transition}: establish existence of the critical
profiles $G_\varsigma$, determine them explicitly or implicitly, and decide whether
the crossover from $\si$ to $\si_{\bRP^1}$ is monotone in $\varsigma$.

\item Establish a quantitative rate in Theorem~\ref{th:GUEmain}.  The heuristic of
Remark~\ref{rem:rate} suggests an $O(n^{-1}\log n)$ first correction in suitable
weak norms, with an explicit profile in $Y$.

\item Beyond the intensity measure: describe the fluctuations of $\mu_n$ around
$\si$, i.e.\ the limit of $N(\int f\,d\mu_n-\int f\,d\si)$. By \eqref{eq:key} this
is a question about fluctuations of logarithmic energies of elliptic ensembles.

\item Characterize the axisymmetric measures on $\bCP^1$ that arise as limiting
level-crossing laws of Gaussian pencils.  In the smooth case the master formula
gives measures of the form
$\bigl(1+2\partial_Y[(1-Y^2)G'(Y)]\bigr)d\si$, while nonsmooth profiles may also
produce singular latitude masses.  Which energy profiles $G$ are admissible?

\item Higher-parameter families: for a generic two-parameter pencil
$\alpha A+\beta B+\gamma C$, the double-eigenvalue discriminant is a plane curve in
$\bCP^2$, while triple degeneracies occur only at a higher-codimension (generically
finite) subset.  Is there an analogue of Theorem~\ref{th:symmetry} for the random
discriminant curve and for its singular points?
\end{enumerate}

\medskip
\begin{ack}
The author thanks Professors Konstantin Zarembo and Andr\'e Galligo for numerous
discussions. During preparation of this manuscript, OpenAI ChatGPT (GPT-5.6 Sol)
was used for language editing, structural revision, and consistency checking. The
author critically reviewed the resulting text and remains responsible for all
mathematical statements, numerical claims, and conclusions.
\end{ack}

\section*{Data availability}
The numerical data supporting the figures and the simulation code are available
from the author upon request.


\begin{thebibliography}{99}

\bibitem[AGZ]{AGZ} Anderson, G.~W., Guionnet, A., Zeitouni, O.,
\emph{An Introduction to Random Matrices}. Cambridge Studies in Advanced
Mathematics 118, Cambridge University Press, Cambridge, 2010.

\bibitem[BC]{BC} Bordenave, C., Chafa\"i, D., Around the circular law,
\emph{Probab. Surv.} \textbf{9} (2012), 1--89.

\bibitem[BYY]{BYY} Bourgade, P., Yau, H.-T., Yin, J., Local circular law for
random matrices, \emph{Probab. Theory Related Fields} \textbf{159} (2014),
545--595.

\bibitem[FZ]{FeinbergZee} Feinberg, J., Zee, A., Non-Hermitian random matrix
theory: method of Hermitian reduction, \emph{Nuclear Phys. B} \textbf{504} (1997),
579--608.

\bibitem[Fo]{ForresterBook} Forrester, P.~J., \emph{Log-Gases and Random
Matrices}. London Mathematical Society Monographs 34, Princeton University Press,
Princeton, NJ, 2010.

\bibitem[FKS]{FKS} Fyodorov, Y.~V., Khoruzhenko, B.~A., Sommers, H.-J., Almost
Hermitian random matrices: crossover from Wigner--Dyson to Ginibre eigenvalue
statistics, \emph{Phys. Rev. Lett.} \textbf{79} (1997), 557--560.

\bibitem[GP]{GP} Galligo, A., Poteaux, A., Computing monodromy via continuation
methods on random Riemann surfaces, \emph{Theor. Comp. Sci.} \textbf{412} (2011),
1492--1507.

\bibitem[Gi1]{Gi1} Girko, V.~L., The $V$-transformation, \emph{Dokl. Akad. Nauk
Ukr. SSR, Ser. Mat.} \textbf{3} (1982), 5--6 (in Russian).

\bibitem[Gi2]{Gi2} Girko, V.~L., The circular law, \emph{Theory Probab. Appl.}
\textbf{29} (1984), 694--706.

\bibitem[Ka]{Ka} Kato, T., \emph{Perturbation Theory for Linear Operators}.
Classics in Mathematics, Springer-Verlag, Berlin, 1995.

\bibitem[Me]{Me} Mehta, M.~L., \emph{Random Matrices}, third edition. Pure and
Applied Mathematics 142, Elsevier/Academic Press, Amsterdam, 2004.

\bibitem[AFS]{AFS} Akemann, G., Fyodorov, Y.~V., Savin, D.~V.,
Spectral density and eigenvector nonorthogonality in complex symmetric random
matrices, arXiv:2511.21643 (2025).

\bibitem[CM]{CM} Crumpton, M.~J., Mezzadri, F., Interpolating non-Hermitian
universality classes A and $AI^\dagger$: eigenvalue density and transition regime,
arXiv:2606.03447 (2026).

\bibitem[ShZa]{ShZa1} Shapiro, B., Zarembo, K., Level crossing in random matrices.
I. Random perturbation of a fixed matrix, \emph{J. Phys. A} \textbf{50} (2017),
no.~4, 045201.

\bibitem[GrShZa]{ShZa2} Gr\o sfjeld, T., Shapiro, B., Zarembo, K., On level
crossing in random matrix pencils. II. Random perturbation of a random matrix,
\emph{J. Phys. A} \textbf{52} (2019), no.~21, 214001.

\bibitem[SCS]{SCS} Sommers, H.-J., Crisanti, A., Sompolinsky, H., Stein, Y.,
Spectrum of large random asymmetric matrices, \emph{Phys. Rev. Lett.} \textbf{60}
(1988), 1895--1898.

\bibitem[TV]{TV} Tao, T., Vu, V., Random matrices: universality of ESDs and the
circular law, \emph{Ann. Probab.} \textbf{38} (2010), 2023--2065.

\bibitem[TV2]{TVnonHerm} Tao, T., Vu, V., Random matrices: universality of local
spectral statistics of non-Hermitian matrices, \emph{Ann. Probab.} \textbf{43}
(2015), 782--874.

\bibitem[ZVW]{ZVW} Zirnbauer, M.~R., Verbaarschot, J.~J.~M., Weidenm\"uller,
H.~A., Destruction of order in nuclear spectra by a residual GOE interaction,
\emph{Nuclear Phys. A} \textbf{411} (1983), 161--180.

\end{thebibliography}
\end{document}